\documentclass[UKenglish]{article}
\usepackage[margin=1in]{geometry}
\usepackage{hyperref}
\usepackage[noblocks]{authblk}
\title{Unfairly Splitting Separable Necklaces}

\author{Patrick Schnider}
\affil[1]{Department of Computer Science, ETH Zürich, Switzerland}
\author[1]{Linus Stalder}
\author[1]{Simon Weber}
\date{}

\usepackage[utf8]{inputenc}
\usepackage{amsmath,amssymb,amsfonts,mathrsfs,amsthm}
\usepackage{mathtools}
\usepackage{thmtools}
\usepackage{thm-restate}
\usepackage[capitalise]{cleveref}
\usepackage{placeins}
\usepackage{url}

\newtheorem{theorem}{Theorem}

\newtheorem{corollary}[theorem]{Corollary}
\newtheorem{lemma}[theorem]{Lemma}
\newtheorem{proposition}[theorem]{Proposition}

\newtheorem{observation}[theorem]{Observation}
\newtheorem{claim}[theorem]{Claim}

\theoremstyle{definition}
\newtheorem{definition}[theorem]{Definition}

\usepackage{enumerate}

\usepackage{tikz}
\usepackage{xspace}

\usepackage{listings}

\newcommand{\f}[1]{\relax\ifmmode#1\else{$#1$}\fi}

\newcommand{\dimension}{\f{n}\xspace}
\newcommand{\Colors}{\f{C}\xspace}

\newcommand{\Reals}{\ensuremath{\mathbb{R}}\xspace}

\newcommand{\UnfairSplitting}{\textsc{$\alpha$-Necklace-Splitting}\xspace}

\newcommand{\nodeA}{\f{a}\xspace}
\newcommand{\nodeB}{\f{b}\xspace}
\newcommand{\nodeC}{\f{c}\xspace}

\newcommand{\bound}{\omega}

\usepackage{tikz}
\usepackage{subcaption}
\usepackage[ruled]{algorithm}
\usepackage{algpseudocodex}

\algrenewcommand\algorithmicrequire{\textbf{Input:}}
\algrenewcommand\algorithmicensure{\textbf{Output:}}

\newcommand{\NP}{\ensuremath{\mathsf{NP}}\xspace}
\newcommand{\coNP}{\ensuremath{\mathsf{co}}-\ensuremath{\mathsf{NP}}\xspace}
\newcommand{\UEOPL}{\ensuremath{\mathsf{UEOPL}}\xspace}
\newcommand{\PPA}{\ensuremath{\mathsf{PPA}}\xspace}

\DeclarePairedDelimiter\abs{\lvert}{\rvert}

\renewcommand{\epsilon}{\ensuremath\varepsilon}

\renewcommand{\phi}{\ensuremath{\varphi}}

\algrenewcommand\algorithmicrequire{\textbf{Input}}
\algrenewcommand\algorithmicensure{\textbf{Returns}}

\newcommand{\trav}{\ensuremath{\mathrm{trav}}}

\newcommand{\tw}{\ensuremath{\mathsf{tw}}}

\newcommand{\nn}{\mathbf{n}}

\newcommand{\pp}{\mathbf{p}}

\newcommand{\fs}{\mathbf{s}}

\newcommand{\xx}{\mathbf{x}}

\begin{document}

\maketitle

\begin{abstract}
The Necklace Splitting problem is a classical problem in combinatorics that has been intensively studied both from a combinatorial and a computational point of view.
It is well-known that the Necklace Splitting problem reduces to the discrete Ham Sandwich problem. This reduction was crucial in the proof of \PPA-completeness of the Ham Sandwich problem. Recently, Borzechowski, Schnider and Weber [ISAAC'23] introduced a variant of Necklace Splitting that similarly reduces to the $\alpha$-Ham Sandwich problem, which lies in the complexity class \UEOPL but is not known to be complete. To make this reduction work, the input necklace is guaranteed to be \emph{$n$-separable}. They showed that these necklaces can be fairly split in polynomial time and thus this subproblem cannot be used to prove \UEOPL-hardness for $\alpha$-Ham Sandwich. We consider the more general \emph{unfair} necklace splitting problem on $n$-separable necklaces, i.e., the problem of splitting these necklaces such that each thief gets a desired fraction of each type of jewels. This more general problem is the natural necklace-splitting-type version of $\alpha$-Ham Sandwich, and its complexity status is one of the main open questions posed by Borzechowski, Schnider and Weber. We show that the unfair splitting problem is also polynomial-time solvable, and can thus also not be used to show \UEOPL-hardness for $\alpha$-Ham Sandwich.
\paragraph{Acknowledgements} Simon Weber is supported by the Swiss National Science Foundation under project no. 204320.
\end{abstract}

\section{Introduction}
One of the most famous theorems in fair division is the \emph{Ham Sandwich theorem}~\cite{HS}. It states that any $d$ point sets in $\Reals^d$ can be simultaneously \emph{bisected} by a single hyperplane.
The Ham Sandwich theorem is closely related to another fair division theorem that lives in \Reals; the \emph{Necklace Splitting theorem}. It states that given $n$ point sets in $\Reals$ (a \emph{necklace} with $n$ types of \emph{jewels}), we can split the real number line at $n$ points such that when we partition the resulting pieces alternatingly, each of the two parts contains exactly half of the jewels of each type. In fact, the Necklace Splitting theorem can be proven by lifting the necklace to the \emph{moment curve} in $\Reals^n$, which is the curve parameterised by $(t,t^2,t^3,\ldots,t^n)$, and then applying the Ham Sandwich theorem.

Under some additional assumptions on the input points, the Ham Sandwich theorem can be significantly strengthened: If the input point sets are \emph{well-separated} and in general position, the \emph{$\alpha$-Ham Sandwich theorem}~\cite{originalDiscreteAlphaHS} says that we can not only simultaneously \emph{bisect} each point set, but we can for each $i$ choose any number $1\leq\alpha_i \leq |P_i|$, and find a single hyperplane that cuts off exactly $\alpha_i$ points from each point set $P_i$. Informally, a family of point sets is well-separated if the union of any subfamily can be separated from the union of the complement subfamily by a single hyperplane. Borzechowski, Schnider and Weber~\cite{nseparableNecklaces} introduced an analogue of this well-separation condition for necklaces: A necklace is \emph{$n$-separable}, if any subfamily can be separated from the complement subfamily by splitting the necklace at at most $n$ points. It is shown in \cite{nseparableNecklaces} that a necklace is $n$-separable if and only if its lifting to the moment curve is well-separated.

Existence theorems such as the Ham Sandwich and the Necklace Splitting theorem can also be viewed from the lens of computational complexity. The corresponding problems are \emph{total search problems}, i.e., problems in which a solution is always guaranteed to exist, but the task is to actually find a solution. In this setting, the strategy used above to prove the Necklace Splitting theorem using the Ham Sandwich theorem also yields a \emph{reduction} from the Necklace Splitting problem to the Ham Sandwich problem. This reduction was crucial for establishing the \PPA-hardness of the Ham Sandwich problem, which is since known to be \PPA-complete~\cite{HamSandwichPPAComplete}. For the $\alpha$-Ham Sandwich problem membership in the subclass $\UEOPL\subseteq \PPA$ is known \cite{aHSinUEOPL}, but no matching hardness. It is thus natural to ask whether \UEOPL-hardness of $\alpha$-Ham Sandwich could be proven by reduction from a Necklace Splitting problem on $n$-separable necklaces. Borzechowski, Schnider and Weber~\cite{nseparableNecklaces} showed that the classical (\emph{fair}) Necklace Splitting problem on $n$-separable necklaces is polynomial-time solvable and thus very unlikely to be \UEOPL-hard. However, the natural necklace-splitting-type analogue of $\alpha$-Ham Sandwich would actually allow for \emph{unfair} splittings on $n$-separable necklaces, where the first of the two thieves should get exactly $\alpha_i$ jewels of type $i$, for some input vector $\alpha$.
We settle the complexity status of this problem variant by providing a polynomial-time algorithm. This completely disqualifies Necklace Splitting variants on $n$-separable necklaces as possible problems to prove \UEOPL-hardness of $\alpha$-Ham Sandwich.

\subsection{Results}
Before we can state our results we have to begin with the most crucial definitions.

\begin{definition}[Necklace]
    A \emph{necklace} is a family $C=\{C_1,\ldots,C_n\}$ of disjoint finite point sets in $\Reals$. The sets $C_i$ are called \emph{colours}, and each point $p\in C_i$ is called a \emph{bead} of colour $i$.
\end{definition}

We align the following definition of $\alpha$-cuts in necklaces as closely as possible with the definition of $\alpha$-cuts in the $\alpha$-Ham Sandwich theorem (which we will define later). We require that an $\alpha$-cut of a necklace splits the necklace at exactly one bead per colour, called the \emph{cut point}. Furthermore, the parity of the permutation of how these cut points appear along \Reals determines whether the first part of the necklace is given to the first or the second thief.

\begin{definition}[$\alpha$-Cut]\label{def:alpha-cut}
    Let $C=\{C_1,\ldots,C_n\}$ be a necklace. A set of $n$ cut points $s_1,\ldots,s_n$ such that for all $i$ we have $s_i\in C_i$ defines the subset $C^+$ of $C_1\cup\ldots\cup C_n$ that is assigned to the first thief as follows. Adding the points $s_0:=-\infty$ and $s_{n+1}:=\infty$, we get the two sets of intervals $A_{\text{even}}:=\{[s_i,s_{i+1}]\;\vert\;0\leq i\leq n, n \text{ even}\}$ and $A_{\text{odd}}:=\{[s_i,s_{i+1}]\;\vert\;0\leq i\leq n, n \text{ odd}\}$. Let $s_{\pi(1)}<\ldots<s_{\pi(n)}$ be the sorted order of the $s_i$. Let $A^+$ be $A_{sgn(\pi)}$, i.e., the set of intervals corresponding to the parity of $\pi$. Then, $\{s_1,\ldots,s_n\}$ is an \emph{$\alpha$-cut} of $C$ for the vector $\alpha=(\alpha_1,\ldots,\alpha_n)$, where $\alpha_i = |A^+\cap C_i|$.
\end{definition}

With this definition, $n$-separability of the necklace and the $\alpha$-Ham Sandwich theorem guarantee a unique $\alpha$-cut for every vector $\alpha$ fulfilling $1\leq \alpha_i\leq |C_i|$, as we will see later. We are now ready to define the \UnfairSplitting problem.

\begin{definition}[\UnfairSplitting]\label{def:UnfairSplitting}
Given an $n$-separable necklace \Colors with \dimension colours, and a vector $\alpha=(\alpha_1,\ldots,\alpha_n)$ with $1\leq \alpha_i\leq |C_i|$, the \UnfairSplitting problem is to find the unique $\alpha$-cut of $C$.
\end{definition}

We are now ready to introduce our results.

\begin{restatable}{theorem}{thmAlgo}\label{thm:mainAlgo}
    \UnfairSplitting is polynomial-time solvable.
\end{restatable}

We contrast this result by showing that without the promise of $n$-separability, the associated decision problem is $\NP$-complete.

\begin{restatable}{theorem}{thmHardness}\label{thm:NPhardness}
    Given a necklace $C$ with \dimension colours, and a vector $\alpha=(\alpha_1,\ldots,\alpha_n)$ with $1\leq \alpha_i\leq |C_i|$, the problem of deciding whether $C$ has an $\alpha$-cut is \NP-complete.
\end{restatable}

Note that Borzechowski, Schnider and Weber~\cite{nseparableNecklaces} proved that it is possible to check whether a given necklace is $n$-separable in polynomial time. This is in contrast to the Ham Sandwich setting, where it has been shown that checking well-separation is \coNP-complete~\cite{WellSeparationCoNP}.

\subsection{Related Work}
Fair necklace splitting was introduced by Alon, Goldberg and West \cite{Alon1987, Alon1986, Goldberg1985}. The problem has been extended to higher dimensions, larger numbers of thieves~\cite{Alon1987,Alon1986}, as well as the continuous setting (\emph{consensus halving}) where necklaces are sets of measures instead of sets of points~\cite{Brams1995, Deng2012, Hobby1965, SegalHalevi2020, Simmons2003}.

A closely related problem is the \emph{paint shop problem}. In the language of necklaces, it asks to split a given necklace among $k\geq 2$ thieves such that each thief gets a desired number of jewels of each type, using as few cuts as possible~\cite{Bonsma2006, Epping2004, Meunier2009}.
This problem can be viewed as the minimisation variant of a generalisation of our \UnfairSplitting problem, without the promise of $n$-separability and possibly more thieves.

\subsection{Proof Techniques}
The algorithm of Borzechowski, Schnider and Weber~\cite{nseparableNecklaces} relies on the following core observation: If some colour only appears in the necklace in two \emph{components} (i.e., consecutive intervals of $\Reals$ that contain only points of this colour), the smaller of the two components can be discarded since the cut point of that colour must lie in the larger component. While for fair splitting this follows quite immediately --- the larger component must be split, otherwise the cut cannot be fair --- this observation does not generalise to unfair splitting. We thus have to use a more complicated approach.

We first show that under the promise of $n$-separability the number of colours that consist of more than two consecutive components is bounded by a constant. For each of these colours we can guess in which component the cut point of this colour lies. For each guess we reduce each of these colours to the single component containing the cut point, and try to compute the $\alpha$-cut for the resulting necklace. For one of these guesses, the resulting $\alpha$-cut must be adaptable to an $\alpha$-cut of the original necklace. To compute the $\alpha$-cut for these necklaces (in which now every colour consists of at most two components), we reduce the necklace further according to some reduction rules similar to those in \cite{nseparableNecklaces}. This process will eventually come to an end; we will reach an \emph{irreducible} necklace. Here comes the crucial part of our proof: We will show that for each $n$, the irreducible necklaces with $n$ colours all have the same \emph{walk graph} (as introduced in \cite{nseparableNecklaces} and below in \Cref{sec:preliminaries}). We translate \UnfairSplitting on irreducible necklaces into an integer linear program (ILP), and use the rigid structure of irreducible necklaces to show that the \emph{primal graph} of this ILP has constant treewidth. Using the FPT algorithm of Jansen and Kratsch~\cite{ILPtreewidth}, we can thus solve these ILPs efficiently.

An experimental implementation of our algorithm for \UnfairSplitting can be found at the following link: 
\url{https://github.com/linstald/alpha-necklace-splitting}.
Note that the implementation does not use the FPT algorithm of Jansen and Kratsch, but a generic ILP solver.

For the \NP-hardness proof of the decision version of \UnfairSplitting we use a reduction from \textsc{e3-Sat}. An implementation of the reduction can also be accessed at the above link.

\section{Preliminaries}\label{sec:preliminaries}
Let us first formally introduce separability, as defined by Borzechowski, Schnider and Weber~\cite{nseparableNecklaces}.
\begin{definition}[Separability]
A necklace \Colors is \emph{$k$-separable} if for all $A \subseteq \Colors$ there exist $k$ \emph{separator points} $s_1<\ldots<s_k\in\Reals$ that separate $A$ from $\Colors \setminus A$. More formally, if we alternatingly label the intervals $(-\infty,s_1],[s_1,s_2],\ldots,[s_{k-1},s_k],[s_k,\infty)$ with $A$ and $\overline{A}$ (starting with either $A$ or $\overline{A}$), for every interval $I$ labelled $A$ we have $I\cap \bigcup_{c\in (\Colors\setminus A)}c=\emptyset$ and for every interval $I'$ labelled $\overline{A}$ we have $I'\cap \bigcup_{c\in A} c = \emptyset$.

The \emph{separability} $sep(\Colors)$ of a necklace $\Colors$ is the minimum integer $k\geq 0$ such that $\Colors$ is $k$-separable.
\end{definition}

We call each maximal set of consecutive points that have the same colour $c$ a \emph{component} of $c$. 
We say a colour $c$ is an \emph{interval}, if it consists of exactly one component. 
In other words, a colour $c$ is an interval if its convex hull does not intersect any other colour $c'$. In \cref{fig:2separable}, the green colour $c$ is an interval, whereas the red colour $a$ is not, it consists of two components.

\begin{figure}[h!]
\begin{subfigure}[b]{0.3\textwidth}
\centering
\includegraphics[page=1]{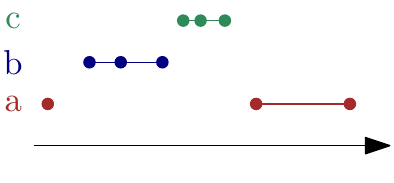}
\caption{``\nodeA \nodeB \nodeC \nodeA'' is 2-separable.}
\label{fig:2separable}
\end{subfigure}
\hfill
\begin{subfigure}[b]{0.3\textwidth}
\centering
\includegraphics[page=2]{figs/separability.pdf}
\caption{``\nodeA \nodeB \nodeA \nodeC'' is 3-separable.}
\label{fig:3separable}
\end{subfigure}
\hfill
\begin{subfigure}[b]{0.3\textwidth}
\centering
\includegraphics[page=3]{figs/separability.pdf}
\caption{``\nodeA \nodeB \nodeA \nodeB \nodeC'' is 4-separable.}
\label{fig:4separable}
\end{subfigure}
\caption{Necklaces with 3 colours \nodeA, \nodeB and \nodeC. The convex hulls of each component are shown.}
\label{fig:exampleOfSeparability}
\end{figure}

Note that for a necklace with $n$ colours, $sep(C)\geq n-1$, and this is tight, as can be seen in \Cref{fig:2separable}.
Borzechowski, Schnider and Weber further showed that $n$-separability is strongly related to the notion of \emph{well-separation}.

\begin{definition}
Let $P_1, \dots, P_k \subset \Reals^d$ be point sets. They are \emph{well-separated} if and only if for every non-empty index set $I \subset [k]$, the convex hulls of the two disjoint subfamilies  $\bigcup_{i \in I} P_i$ and $\bigcup_{i \in [k]\setminus I} P_i$ can be separated by a hyperplane.
\end{definition}

\begin{lemma}[\cite{nseparableNecklaces}]
\label{lem:dSeparable<=>wellSeparable}
Let \Colors be a set of \dimension colours in $\Reals$. Let $\Colors'$ be the set of subsets of $\Reals^\dimension$ obtained by lifting each point in each colour of $\Colors$ to the \dimension-dimensional moment curve using the function $f(t)=(t,t^2,\ldots,t^\dimension)$.
Then the set \Colors is \dimension-separable if and only if $\Colors'$ is well-separated.
\end{lemma}

The following theorem due to Steiger and Zhao~\cite{originalDiscreteAlphaHS} shows that we can always unfairly bisect well-separated point sets.
\begin{lemma}[{$\alpha$-Ham-Sandwich Theorem, \cite{originalDiscreteAlphaHS}}]\label{lem:alphaHS}
Let $P_1, \dots, P_\dimension \subset \Reals^\dimension$ be finite well-separated point sets in general position, and let $\alpha_1, \dots, \alpha_\dimension$ be positive integers with $\alpha_i \leq \abs{P_i}$, then there exists a unique $(\alpha_1, \dots, \alpha_\dimension)$-cut, i.e., a hyperplane $H$ that contains a point from each colour and such that for the closed positive halfspace $H^+$ bounded by $H$ we have $\abs{H^+ \cap P_i} = \alpha_i$. Here, the positive side of a hyperplane $H$ containing one point $p_i$ per point set $P_i$ is determined by the orientation of these $p_i$, i.e., for any point $h\in H^+$ the simplex $(p_1,\ldots,p_n,h)$ is oriented positively.
\end{lemma}

Through the  classical reduction of Necklace Splitting to the Ham-Sandwich problem obtained by lifting the points to the moment curve, as it appeared in many works before~\cite{momentCurve2,HamSandwichPPAComplete,matousek2002lectures,momentCurve1}, we can easily obtain the following theorem; by \cref{lem:dSeparable<=>wellSeparable} the point sets lifted to the moment curve are well-separated, and thus \Cref{lem:alphaHS} applies.
\begin{theorem}
\label{thm:uniquenessOfSolution}
\UnfairSplitting always has a unique solution.
\end{theorem}

To argue about the separability of necklaces, we use the view of \emph{walk graphs} that were also introduced in~\cite{nseparableNecklaces}.
\begin{definition}[Walk graph]
Given a necklace \Colors, the walk graph $G_\Colors$ is the multigraph with $V=\Colors$ and with every potential edge $\{a,b\}\in \binom{V}{2}$ being present with the multiplicity equal to the number of pairs of points $p\in a,p'\in b$ that are neighbouring.
\end{definition}

The walk graphs of the example necklaces in \cref{fig:exampleOfSeparability} can be seen in \cref{fig:WalkGraphExample}.

\begin{figure}[h!]
\begin{subfigure}[b]{0.3\textwidth}
\centering
\begin{tikzpicture}
\node (a) at (0,0) {\nodeA};
\node (b) at (2,0) {\nodeB};
\node (c) at (1,1.5) {\nodeC};
\draw (a) edge (b);
\draw (b) edge (c);
\draw (c) edge (a);
\end{tikzpicture}
\caption{Walk graph for ``\nodeA \nodeB \nodeC \nodeA''}
\end{subfigure}
\begin{subfigure}[b]{0.3\textwidth}
\centering
\begin{tikzpicture}
\node (a) at (0,0) {\nodeA};
\node (b) at (2,0) {\nodeB};
\node (c) at (1,1.5) {\nodeC};
\draw (a) edge (b);
\draw (a) edge[bend left] (b);
\draw (c) edge (a);
\end{tikzpicture}
\caption{Walk graph for ``\nodeA \nodeB \nodeA \nodeC''}
\end{subfigure}
\begin{subfigure}[b]{0.3\textwidth}
\centering
\begin{tikzpicture}
\node (a) at (0,0) {\nodeA};
\node (b) at (2,0) {\nodeB};
\node (c) at (1,1.5) {\nodeC};
\draw (a) edge (b);
\draw (a) edge[bend left] (b);
\draw (a) edge[bend right] (b);
\draw (c) edge (b);
\end{tikzpicture}
\caption{Walk graph for ``\nodeA \nodeB \nodeA \nodeB \nodeC''}
\end{subfigure}
\caption{Walk graphs of the examples in \cref{fig:exampleOfSeparability}.}
\label{fig:WalkGraphExample}
\end{figure}
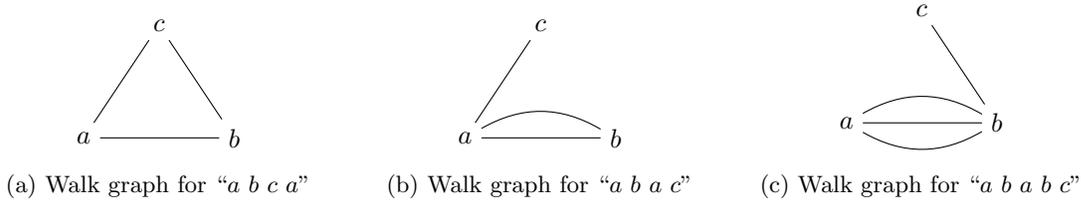
Note that given a necklace \Colors as a set of point sets, the walk graph can be built in linear time in the size of the necklace $N:=\sum_{c\in\Colors}|c|$ .

Recall that a graph is \emph{Eulerian} if it contains a Eulerian tour, a closed walk that uses all edges exactly once. A graph is \emph{semi-Eulerian} if it contains a Eulerian path, a (not necessarily closed) walk that uses all edges exactly once.

\begin{observation}\label{obs:semiEulerian}
The walk graph of a necklace is connected and semi-Eulerian, and thus at most two vertices have odd degree.
\end{observation}

The separability of a necklace turns out to be equivalent to the max-cut in its walk graph.

\begin{definition}[Cut]
    In a (multi-)graph $G$ on the vertices $V$, a \emph{cut} is a subset $A\subseteq V$. The \emph{size $\mu(A)$ of a cut} $A$ is the number of edges $\{u,v\}$ in $G$ such that $u\in A$ and $v\not\in A$. The \emph{max-cut}, denoted by $\mu(G)$, is the largest size of any cut $A\subseteq V$. 
\end{definition}
\begin{lemma}[\cite{nseparableNecklaces}]
    For every necklace $\Colors$, we have $sep(\Colors)=\mu(G_\Colors)$.
\end{lemma}

In our proofs we often need to show that certain structures or properties do not appear in walk graphs of necklaces with bounded separability. The general strategy for these proofs is to show that walk graphs with these structures or properties have a large max-cut, and thus the corresponding necklaces cannot have the claimed separability. Our main tool for this is the following bound that is a corollary of a theorem of Poljak and Turzík.
\begin{corollary}[\cite{poljakBoundForNonsimple}]\label{cor:multigraphErdos}
A connected (multi-)graph $G$ with $n$ vertices and $m$ edges has a maximum cut $\mu(G)$ of at least $\bound(G):=\frac{m}{2}+\frac{n-1}{4}$.
\end{corollary}

Using \Cref{cor:multigraphErdos} we immediately get a bound on the number of edges in the walk graph of a $(n-1+\ell)$-separable necklace.

\begin{lemma}
    Let $C$ be a $(n-1+\ell)$-separable necklace with $n$ colours for $n \geq 8$.
    Then the number of edges $m$ in its walk graph $G_C$ is at most $m \leq \frac{3n-3}{2}+2\ell$.
    \label{lem:num-edges}
\end{lemma}
\begin{proof}
    We have $\mu(G_C) \geq \frac{m}{2} + \frac{n-1}{4}$ and $\mu(G_C) \leq n-1+\ell$.
    Therefore, we have
    \begin{align*}
              && n -1 + \ell &\geq \frac{m}{2} + \frac{n-1}{4} &&\\
        &\iff&  4n -4 + 4\ell  &\geq 2m + n-1 \\
        &\iff&  \frac{3n-3}{2} + 2\ell  &\geq m. &&\hfill\qedhere
    \end{align*}
\end{proof}

We also immediately get a convenient bound on the number of intervals.

\begin{lemma}
    Let $C$ be a $(n-1+\ell)$-separable necklace with $n$ colours for $n \geq 8$.
    Then the numbers of intervals $k$ is at least $k \geq \frac{n+1}{2}-2\ell$.
    \label{lem:num-intervals}
\end{lemma}
\begin{proof}
    Every interval has degree 2 (if it is not at the start or end of the necklace).
    Moreover, every colour that is not an interval has degree at least 4 (if it is not at the start or end of the necklace).
    We get:
    \begin{align*}
        \sum_{c \in C} \deg(c) &\geq k\cdot 2 + (n-k)\cdot 4 - 2 \\
                               & = 4n - 2k -2
    \end{align*} 
    From \Cref{lem:num-edges} we have $\sum_{c \in C} \deg(c) \leq 3n -3 + 4\ell$.
    Hence, we must have $k \geq \frac{n+1}{2}-2\ell$.
\end{proof}

\section{Tractability of the Search Problem}
In this section we prove \Cref{thm:mainAlgo} by providing a polynomial-time algorithm for \UnfairSplitting. As mentioned above, the algorithm works in two main phases.

In the first phase, the necklace is first reduced to a necklace where each colour consists of at most two components by guessing the correct component to cut for colours with three or more components. The necklace with only colours with at most two components is then further reduced to an \emph{irreducible} necklace.

In the second phase, we reduce necklace splitting in an irreducible necklace to a labelling problem of its walk graph. This labelling problem is then modelled as an integer linear program, which turns out to be tractable in polynomial time. To prove that this ILP is tractable, we prove some strong structural properties about the walk graphs of irreducible necklaces.

\subsection{Reducing Necklaces}
Instead of solving \UnfairSplitting, we will actually solve the following slightly more general problem. Given a vector $\alpha=(\alpha_1,\ldots,\alpha_n)$, we define the complement vector $\overline{\alpha}$ as the vector $(|C_1|-\alpha_1+1,\ldots,|C_n|-\alpha_n+1)$. If $\alpha$ denotes the number of points per point set on the positive side of a cut, $\overline{\alpha}$ denotes the number of points on the other side, both sides including the cut points. Since the cut parity can change in our reduction steps and thus the positive side may become the negative side and vice versa, the following problem is nicer to solve recursively than \UnfairSplitting.

\begin{definition}$\alpha$-$\overline{\alpha}$-\textsc{Necklace-Splitting}

\begin{tabular}{ll}
\textbf{Input:}  & An $n$-separable necklace $C = \{C_1, \dots, C_n\}$, \\ 
                    & a vector $\alpha = (\alpha_1, \dots, \alpha_n)$, $\alpha_i \in \{1, \dots, |C_i|\}$.\\
\textbf{Output:} & A pair $(S, \overline{S})$, where $S$ is the unique $\alpha$-cut and $\overline{S}$ the unique $\overline{\alpha}$-cut.
\end{tabular}
\end{definition}

\subsubsection{Reducing to at most two components per colour}\label{sec:firstphase}
In this section we will algorithmically reduce $\alpha$-$\overline{\alpha}$-\textsc{Necklace-Splitting} to the following subproblem, where each colour in the necklace consists of at most two components.
\begin{definition}$\alpha$-$\overline{\alpha}$-\textsc{Necklace-Splitting$_2$}

\begin{tabular}{ll}
\textbf{Input:}  & An $n$-separable necklace $C = \{C_1, \dots, C_n\}$, where each colour has at most 2\\
                 & components, and a vector $\alpha = (\alpha_1, \dots, \alpha_n)$, $\alpha_i \in \{1, \dots, |C_i|\}$.\\
\textbf{Output:} & A pair $(S, \overline{S})$, where $S$ is the unique $\alpha$-cut and $\overline{S}$ the unique $\overline{\alpha}$-cut.
\end{tabular}
\end{definition}

Our reduction is based on the following observation.
\begin{lemma}
    Let $C$ be an $n$-separable necklace with $n$ colours for $n \geq 8$.
    Then in the necklace at least one of the following must be true:
    \begin{enumerate}[(i)]
        \item there are two neighbouring intervals, or
        \item there is no colour with more than four components, and at most two colours have more than two components.
    \end{enumerate}
    \label{lem:atmost4components}
\end{lemma}
\begin{proof}
    We show that when both (i) and (ii) do not hold we get a contradiction.
    Let $G_C$ be the walk graph of $C$.
    
    Assuming that there are no two neighbouring intervals, taking the set of intervals $I$ as a cut yields a cut of size $2|I| - 2$ --- we need to subtract 2 since possibly one interval at the start and one at the end of the necklace only have one incident edge.
    Therefore, we have $\mu(G_C) \geq 2|I|-2$ and thus $|I| \leq \frac{\mu(G_C) + 2}{2}$.
    Since $C$ is $n$-separable, we have $\mu(G_C) \leq n$ which implies $|I| \leq \frac{n+2}{2}$.

    Now assuming that (ii) does not hold either, there is either a colour with at least five components, or more than two colours with at least three components.

    In the first case, there are $|I|$ colours with one component, at least one colour with five components and all other colours have at least two components.
    Since the first and last component both lack one edge for this calculation, we need to subtract $2$ edges for these.
    We thus get a lower bound on the sum of degrees in the walk graph.
    \begin{align*}
        \sum_{c\in C} \deg(c) &\geq |I| \cdot 2 + (n-|I|-1)\cdot 4 + 10 - 2 \\
        &= 4n - 2|I| + 4\\
        &\geq 4n - n - 2 + 4\\
        &= 3n + 2
    \end{align*}
    However, by \Cref{lem:num-edges} we have that $\sum_{c\in C} \deg(c) \leq 3n+1$.
    Therefore, the above result is a contradiction.
   
    In the second case, there are $|I|$ colours with one component, at least three colours with three components, and all other colours have at least two components. Again subtracting $2$ edges for the start and end we get the bound
    \begin{align*}
        \sum_{c \in C}\deg(c) &\geq |I|\cdot 2 + 3\cdot6 + (n-|I|-3)\cdot4 - 2\\
                              &= 4n - 2|I| + 4 \\
                              &\geq 4n - n - 2 + 4 \\
                              &= 3n + 2,
    \end{align*}
    which again contradicts \Cref{lem:num-edges}.

    We conclude that either condition (i) or (ii) has to hold.
\end{proof}
Algorithmically, the conditions (i) and (ii) can be used as follows.
If there are two neighbouring intervals, since an $\alpha$-cut must go through exactly one bead of each colour, there must be a cut in both intervals.
Moreover, in between these cuts there is no other cut point.
Therefore, we remove the two intervals and solve on the newly obtained necklace recursively.
Finally, we add the cuts at the right positions in the removed intervals.

If there are no neighbouring intervals, condition (ii) states that at most two colours have more than two components.
The strategy will be to look at these colours and test out every possible component per colour. Again by condition (ii), this requires at most $4\cdot 4=16\in O(1)$ tests. That is, for each colour with more than two components we fix a component and remove all other components of that colour.
In this smaller necklace (that still consists of $n$ colours) we recursively solve for an $\alpha$-cut.
The necklace for the recursive call will then be a necklace where each colour has at most two components, so we need to solve an $\alpha$-$\overline{\alpha}$-\textsc{Necklace-Splitting$_2$} instance in the recursive step.

Since there must be an $\alpha$-cut, one combination of components must lead to a smaller necklace whose $\alpha$-cut can be augmented by inserting the cut of each colour in the fixed component.

In order to see that our recursive calls indeed work, we also need that removing neighbouring intervals yields a $(n-2)$-separable necklace on $n-2$ colours, and removing all but one component from a colour does not destroy $n$-separability either. Both of these facts are shown in \cite[Lemmas 19 and 20]{nseparableNecklaces}.

This lets us conclude the following proposition.
\begin{proposition}
    Let $T_2(n, N)$ be the time to solve $\alpha$-$\overline{\alpha}$-\textsc{Necklace-Splitting$_2$} on an $n$-separable necklace with $n$ colours and a total of $N$ beads.
    Then $\alpha$-$\overline{\alpha}$-\textsc{Necklace-Splitting} on an $n$-separable necklace with a total of $N$ beads and $n$ colours can be solved in at most
    \(
        \mathcal{O}(T_2(n, N) + n \cdot N)
    \)
    time.
    \label{prop:reduce-to-two-components}
\end{proposition}
\begin{proof}
    We use \Cref{alg:first-phase} for $\alpha$-$\overline{\alpha}$-\textsc{Necklace-Splitting}.
    The algorithm is recursive and thus, if $n$ is small enough we can apply a brute force algorithm that iterates through every possible cut and checks if it has found the $\alpha$-cut or $\overline{\alpha}$-cut.
    Otherwise, the algorithm proceeds as follows.

    In a first step, neighbouring intervals are removed from the necklace and the $\alpha$- and $\overline{\alpha}$-cut is computed in the obtained necklace recursively.
    Then in these cuts the right cuts are added in the removed intervals to get the $\alpha$- and $\overline{\alpha}$-cut.
    We need to make sure to maintain the cut parity when inserting these cuts, so the algorithm checks first if the cut parity switches when inserting these cuts and if yes, it swaps the $\alpha$- and $\overline{\alpha}$-cuts obtained by the recursive call.
    That is, if the cut parity switches, the algorithm uses the obtained $\overline{\alpha}$-cut and turns it into the $\alpha$-cut by inserting the correct cuts in the removed intervals (and similarly turns the $\alpha$-cut into an $\overline{\alpha}$-cut).

    If there are no neighbouring intervals in the necklace, the algorithm determines the set of all colours with at least three components, call this set $C_3$.
    If $C_3$ is empty, the necklace consists only of colours each having at most two components, so we can use an algorithm for $\alpha$-$\overline{\alpha}$-\textsc{Necklace-Splitting$_2$} to compute the $\alpha$- and $\overline{\alpha}$-cut.

    Otherwise, if $C_3$ is not empty, the algorithm iterates over all combinations of components of colours in $C_3$.
    That is, each iteration considers a choice $(c_{i_1}, \dots, c_{i_{|C_3|}})$ of exactly one component $c_{i_k}$ of each colour $C_{i_k}$ in $C_3$.
    In this iteration, all components from colours in $C_3$ except the components from the current choice are removed.   
    We obtain a necklace that still consists of $n$ colours but each colour has at most two components.
    Therefore, we can use an algorithm for $\alpha$-$\overline{\alpha}$-\textsc{Necklace-Splitting$_2$} to find an $\alpha$- and $\overline{\alpha}$-cut in the new necklace.
    However, since we removed beads from the colours in $C_3$ the $\alpha$ might be invalid for this necklace.
    We therefore use a dummy value of 1 for these colours.
    Then the algorithm checks, if the obtained cuts can be turned to the corresponding $\alpha$- or $\overline{\alpha}$-cut by shifting the cuts along the beads within the components of the current iteration.
    
    \begin{algorithm}
    \caption{$\alpha$-$\overline{\alpha}$-\textsc{Necklace-Splitting}}
    \label{alg:first-phase}
    \begin{algorithmic}[1]
        \Require $C = \{C_1, \dots, C_n\}$ an $n$-separable necklace and $\alpha = (\alpha_1, \dots, \alpha_n)$ with $\alpha_i \in \{1, \dots, |C_i|\}$.
        \Ensure $S, \overline{S}$ where $S$ is the unique $\alpha$-cut and $\overline{S}$ the unique $\overline{\alpha}$-cut of $C$.
        \If {$n \leq 8$}
            \State \Return $\textsc{BruteForce}(C, \alpha)$ \Comment{checks every possible cut}
        \EndIf
        \If {there are neighbouring intervals $C_i, C_j$ in $C$}
            \State $S', \overline{S'} \gets$ $\alpha$-$\overline{\alpha}$-\textsc{Necklace-Splitting}$(C \setminus \{C_i,C_j\}, \alpha \setminus \{\alpha_i, \alpha_{j}\})$
            \label{line:first-phase-recursive1}
            \If {inserting cuts in $C_i, C_j$ switches the cut parity}
                \State $S', \overline{S'} \gets \overline{S'}, S$ 
            \EndIf
            \State $S , \overline{S} \gets S', \overline{S'}$ with the right $\alpha,\overline{\alpha}$-cuts in $C_i,C_j$
            \State \Return $S, \overline{S}$
            \label{line:first-phase-return1}
        \EndIf
        \label{line:first-phase-firstpart}
        \Statex \Comment{There are no neighbouring intervals}
        \State $C_3 \gets \{i \in [n] \mid C_i \text{ has at least three components}\}$
        \If {$|C_3| = 0$}
            \State \Return $\alpha$-$\overline{\alpha}$-\textsc{Necklace-Splitting$_2$}$(C, \alpha)$ 
            \label{line:first-phase-recursive2}
        \EndIf
        \State $S , \overline{S} \gets \emptyset$
        \For {all combinations $(c_{i_1}, \dots, c_{i_{|C_3|}})$ of components $c_{i_k} \subseteq C_{i_k}$ for $i_k \in C_3$}
            \label{line:first-phase-forloop}
            \State $C' \gets C \setminus \{C_{i_1}, \dots, C_{i_{|C_3|}}\} \cup \{c_{i_1}, \dots, c_{i_{|C_3|}}\}$\label{line:firstphaseRestrictToOne}
            \Statex \Comment{$C'$ has exactly one fixed component per colour in $C_3$}
            \Statex \Comment{$C'$ consists only of colours with at most two components}
            \State $\alpha' \gets \alpha$
            \State $\alpha'_i \gets 1$ for all $i \in C_3$ \Comment{Take a dummy value for $\alpha'_i$}
            \State $S', \overline{S'} \gets $$\alpha$-$\overline{\alpha}$-\textsc{Necklace-Splitting$_2$}$(C', \alpha')$ 
            \label{line:first-phase-recursive3}
            \If {$S'$ can be shifted to $\alpha$-cut in $C$ by shifting cuts in $c_{i_1}, \dots, c_{i_{|C_3|}}$}
                \State $S \gets S'$ shifted to $\alpha$-cut in $C$ shifting cuts in $c_{i_1}, \dots, c_{i_{|C_3|}}$
                \label{line:first-phase-assignment}
            \EndIf
            \If {$\overline{S'}$ can be shifted to $\overline{\alpha}$-cut in $C$ by shifting cuts in $c_{i_1}, \dots, c_{i_{|C_3|}}$}
                \State $\overline{S} \gets \overline{S'}$ shifted to $\overline{\alpha}$-cut in $C$ shifting cuts in $c_{i_1}, \dots, c_{i_{|C_3|}}$
            \EndIf
        \EndFor
        \State \Return $S, \overline{S}$
        \label{line:first-phase-return2}
    \end{algorithmic}
    \end{algorithm}
    To show that the algorithm is correct, we need to show that the algorithm returns the correct result in the case when removing neighbouring intervals (line~\ref{line:first-phase-return1}), and in the case when removing the components from colours with at least three components (line~\ref{line:first-phase-return2}).
    Moreover, we need to show that the recursive call when removing neighbouring intervals is valid (line~\ref{line:first-phase-recursive1}), and similarly, the calls to the algorithm for $\alpha$-$\overline{\alpha}$-\textsc{Necklace-Splitting$_2$} is valid (lines~\ref{line:first-phase-recursive2} and \ref{line:first-phase-recursive3}).

    The recursive call when removing neighbouring intervals (line~\ref{line:first-phase-recursive1}) is valid and returns a correct result, as the necklace of the recursive call has $n-2$ colours and is $(n-2)$-separable by \cite[Lemma 19]{nseparableNecklaces}.
    Note that inserting cuts in neighbouring intervals either changes the parity of the cut permutation for all cuts or for no cut.
    This can be seen as moving neighbouring intervals as a pair cannot add additional inversions to the cut permutation.
    Thus, moving them to the beginning of the necklace, there are always either an even or an odd number of inversions, regardless of the cut to augment.
    We can therefore indeed check whether inserting the cuts in the intervals changes parity and if it does we swap the $\alpha$- and $\overline{\alpha}$-cut of the necklace from the recursive call.
    This way, we make sure that when inserting the correct cut points in the removed intervals, the obtained cuts are indeed valid $\alpha$- and $\overline{\alpha}$-cuts.
    This shows that in this case (line~\ref{line:first-phase-return1}) the algorithm indeed returns the unique $\alpha$- and $\overline{\alpha}$-cuts.

    We already argued in the description of the algorithm that the calls to the algorithm for $\alpha$-$\overline{\alpha}$-\textsc{Necklace-Splitting$_2$} (lines~\ref{line:first-phase-recursive2} and \ref{line:first-phase-recursive3}) are valid, since these calls only work on $n$-separable necklaces that still have $n$ colours and additionally, each colour has at most two components.
    Also correctness for the case when there are no colours with at least three components (line~\ref{line:first-phase-recursive2}) follows from the correctness of the algorithm for $\alpha$-$\overline{\alpha}$-\textsc{Necklace-Splitting$_2$}.

    Next we show that the cut returned for the case when there are colours with at least three components (line~\ref{line:first-phase-return2}) is correct.
    We only argue for the $\alpha$-cut, the case for the $\overline{\alpha}$-cut is analogous.
    Notice that in the unique $\alpha$-cut there is exactly one component per colour in $C_3$ where the cut lies in.
    Therefore, there must be one iteration (in the for loop in line~\ref{line:first-phase-forloop}) for exactly that choice of components.
    For that iteration, the cut obtained by the call of the algorithm for $\alpha$-$\overline{\alpha}$-\textsc{Necklace-Splitting$_2$} (in line~\ref{line:first-phase-recursive3}) is a valid $\alpha$-cut in the necklace $C'$ for all colours except the ones in $C_3$, where $C'$ is the necklace obtained by removing all components from colours in $C_3$ except the components from the current choice.
    Then the $\alpha$-cut in $C'$ can be shifted to an $\alpha$-cut in $C$ and the algorithm will correctly remember this $\alpha$-cut (in line~\ref{line:first-phase-assignment}).
    Observe that we do not have to worry about changing cut parities, since the necklace $C'$ works on the same set of colours and the cuts are only shifted within components, so the cut permutation does not change.
    Moreover, since the $\alpha$-cut is unique the algorithm will only remember one $\alpha$-cut.
    Hence, the cut obtained in this case (line~\ref{line:first-phase-return2}) will be the correct $\alpha$-cut.
    
    For the runtime analysis note that the brute force step takes $\mathcal{O}(N)$ time.
    Moreover, as long as there are neighbouring intervals the algorithm takes $\mathcal{O}(N)$ per recursive call.
    In each recursive call the number of colours decreases by 2, so there are at most $\mathcal{O}(n)$ iterations.
    Hence, the runtime for removing neighbouring intervals is at most $\mathcal{O}(n \cdot N)$.

    By \Cref{lem:atmost4components} we have $|C_3| \leq 2$ and each colour in $C_3$ has at most four components.
    Therefore, the number of iterations over components of $C_3$ (in the loop of line~\ref{line:first-phase-forloop}) is constant, where each iteration takes time $\mathcal{O}(T_2(n, N) + N)$.

    Hence, the total runtime is $\mathcal{O}(n \cdot N + T_2(n, N) + N) = \mathcal{O}(T_2(n, N) + n\cdot N)$.
\end{proof}

\subsubsection{Further reductions until irreducibility}
To solve the $\alpha$-$\overline{\alpha}$-\textsc{Necklace-Splitting$_2$} problem we perform some further reductions to reach a necklace that we cannot further reduce using any techniques known to us, which we will call irreducible.

\begin{definition}\label{def:irreducible}
    An $n$-separable necklace $C$ is called \emph{irreducible} if and only if all of the following conditions hold.
    \begin{enumerate}[(i)]
        \item All colours have at most two components,
        \item there are no neighbouring intervals,
        \item neither the first nor the last component is an interval,
        \item the first and the last component are of different colours.
    \end{enumerate}
\end{definition}

We thus reduce $\alpha$-$\overline{\alpha}$-\textsc{Necklace-Splitting$_2$} to the following subproblem.
\begin{definition}$\alpha$-$\overline{\alpha}$-\textsc{Irr-Necklace-Splitting}

\begin{tabular}{ll}
\textbf{Input:}  & An $n$-separable irreducible necklace $C = \{C_1, \dots, C_n\}$,  \\
                 & a vector $\alpha = (\alpha_1, \dots, \alpha_n)$, $\alpha_i \in \{1, \dots, |C_i|\}$.\\
\textbf{Output:} & A pair $(S, \overline{S})$, where $S$ is the unique $\alpha$-cut and $\overline{S}$ the unique $\overline{\alpha}$-cut.
\end{tabular}
\end{definition}

To perform this reduction, we search for violations of any of the conditions (ii) to (iv) in \Cref{def:irreducible}. Note that although we already removed all neighbouring intervals in \Cref{sec:firstphase}, neighbouring intervals may be reintroduced by removing some components from the necklace when removing components from colours on line~\ref{line:firstphaseRestrictToOne} of \Cref{alg:first-phase}, or when dealing with any of the other three violations in the further reduction.

In the following we show how the $\alpha$-cut can be computed for each of the violations to conditions (ii) to (iv). We already know how to deal with violations to condition (ii), i.e., neighbouring intervals, from the previous subsection.

If condition (iii) is violated because the first component is an interval, the idea is to remove that colour and solve recursively on the smaller instance $C'$.
To obtain an $\alpha$-cut we take either the $\alpha$-cut or the $\overline{\alpha}$-cut in $C'$ and add the cut the first component at the correct bead.
Note that if the first component is of colour $C_i$ and the number of colours $C_j \in C'$ such that $j < i$ is odd, then augmenting the cut will switch the parity of the cut permutation, since adding the first cut in $C_i \notin C'$ adds an odd number of inversions to the permutation.
Moreover, adding a cut in the first component flips the positive and the negative side of the rest of the necklace once more. Thus, if the first component is such that the cut permutation parity does \emph{not} flip, we extend the $\overline{\alpha}$-cut of $C'$, and otherwise the $\alpha$-cut.

The same idea is applied if the last component is an interval $C_i$.
We again remove this colour, solve recursively and add a cut in the last component.
In this case, adding the cut to the last component does not additionally flip the positive and negative side of the rest of the necklace, but the permutation flips parity if the number of colours $C_j \in C'$ such that $j > i$ is odd.

It is easy to see that removing an interval at the beginning or the end of the necklace decreases the separability by $1$, since this interval can always be put on the correct side of a cut so one cut is needed to separate it from the rest. Thus, $C'$ is indeed $(n-1)$-separable, and the recursive call is legal.

If condition (iv) is violated because the first and last component are of the same colour, we use a similar idea. We obtain the necklace $C'$ by removing this colour.
Note again that $C'$ is $(n-1)$-separable. We again use recursion to get an $\alpha$- and $\overline{\alpha}$-cut for the remaining colours in $C'$. The $\alpha$-cut in $C$ can now be found by trying to augment both the $\alpha$-cut and the $\overline{\alpha}$-cut of $C'$ by inserting a cut in the first or last component. At least one of the two augmentations must succeed as removing the cut in $C_i$ from the unique $\alpha$-cut in $C$ must yield either the $\alpha$- or the $\overline{\alpha}$-cut of $C'$. Similarly, we can also find the $\overline{\alpha}$-cut of $C$.

For the sake of completeness we give pseudocode for the way these reductions can be implemented in \Cref{alg:second-phase}.
From the arguments above it follows that the algorithm is correct and runs in $\mathcal{O}(T_{irr}(n, N) + n \cdot N)$ time, where $T_{irr}(n, N)$ is the time needed to solve $\alpha$-$\overline{\alpha}$-\textsc{Irr-Necklace-Splitting}.
This proves the following proposition.
\begin{proposition}
    Let $T_{irr}(n, N)$ be the time to solve $\alpha$-$\overline{\alpha}$-\textsc{Irr-Necklace-Splitting} on an $n$-separable irreducible necklace with $n$ colours and a total of $N$ beads.
    Then $\alpha$-$\overline{\alpha}$-\textsc{Necklace-Splitting$_2$} on an $n$-separable necklace with a total of $N$ beads and $n$ colours can be solved in at most
    \(
        \mathcal{O}(T_{irr}(n, N) + n \cdot N)
    \)
    time.
    \label{prop:reduce-to-irreducible}
\end{proposition}

\begin{algorithm}
\caption{$\alpha$-$\overline{\alpha}$-\textsc{Necklace-Splitting$_2$}}
\label{alg:second-phase}
\begin{algorithmic}[1]
    \Require $C = \{C_1, \dots, C_n\}$ an $n$-separable necklace where each colour consists of at most 2 components and $\alpha = (\alpha_1, \dots, \alpha_n)$ with $\alpha_i \in [1, |C_i|]$.
    \Ensure $S, \overline{S}$ where $S$ is the unique $\alpha$-cut and $\overline{S}$ the unique $\overline{\alpha}$-cut of $C$.    
    \If {$n \leq 8$}
        \State \Return $\textsc{BruteForce}(C, \alpha)$ \Comment{checks every possible cut}
    \EndIf
    \If {there are neighbouring intervals $C_i, C_j$ in $C$}
            \State Proceed as in \Cref{alg:first-phase}.
        \EndIf
    \If {first component is an interval}
        \State $C_i \gets \text{ the colour of first component}$
        \State $S', \overline{S'} \gets$ $\alpha$-$\overline{\alpha}$-\textsc{Necklace-Splitting$_2$}$(C \setminus \{C_i\}, \alpha \setminus \{\alpha_i\})$
        \If {adding cut in $C_i$ does not flip cut permutation parity}
            \State $S', \overline{S'} \gets \overline{S'}, S'$
        \EndIf
        \State $S, \overline{S} \gets S', \overline{S'}$ with the right $\alpha,\overline{\alpha}$-cut in $C_i$
        \State \Return $S, \overline{S}$
    \EndIf
    \If {last component is an interval}
        \State $C_i \gets \text{ the colour of last component}$
        \State $S', \overline{S'} \gets$ $\alpha$-$\overline{\alpha}$-\textsc{Necklace-Splitting$_2$}$(C \setminus \{C_i\}, \alpha \setminus \{\alpha_i\})$
        \If {adding cut in $C_i$ does flip cut permutation parity}
            \State $S', \overline{S'} \gets \overline{S'}, S'$
        \EndIf
        \State $S , \overline{S} \gets S', \overline{S'}$ with the right $\alpha,\overline{\alpha}$-cuts in $C_i$
        \State \Return $S, \overline{S'}$
    \EndIf
    \If {first and last component, $c_1, c_2$, are from the same colour $C_i$}
        \State $S', \overline{S'} \gets$ $\alpha$-$\overline{\alpha}$-\textsc{Necklace-Splitting$_2$}$(C \setminus \{C_i\}, \alpha \setminus \{\alpha_i\})$ 
        \If {$S'$ can be shifted to $\alpha$-cut by shifting cut in $c_1$ or in $c_2$}
        \label{line:second-phase-ass1}
            \State $S \gets S'$ shifted to $\alpha$-cut by shifting cut in $c_1$ or in $c_2$
        \EndIf
        \If {$\overline{S'}$ can be shifted to $\alpha$-cut by shifting cut in $c_1$ or in $c_2$}
            \State $S \gets \overline{S'}$ shifted to $\alpha$-cut by shifting cut in $c_1$ or in $c_2$
        \EndIf
        \label{line:second-phase-ass2}
        \State Find the $\overline{\alpha}$-cut $\overline{S}$ analogous to lines~\ref{line:second-phase-ass1}--\ref{line:second-phase-ass2}
        \State \Return $S, \overline{S}$
    \EndIf
    \State \Return $\alpha$-$\overline{\alpha}$-\textsc{Irr-Necklace-Splitting}$(C, \alpha)$ \Comment{ $C$ is irreducible}
    
\end{algorithmic}
\end{algorithm}

\subsection{Structure of Irreducible Necklaces}
To be able to solve $\alpha$-$\overline{\alpha}$-\textsc{Irr-Necklace-Splitting}, we will first analyse the structure of the walk graphs of irreducible necklaces.

First, we claim that the walk graphs $G = (V, E)$ of irreducible necklaces with $n$ colours can be characterised as follows.
\begin{restatable}{lemma}{irreducibleCond}\label{lem:irreducibleconditions}
    A graph $G = (V, E)$ on $n$ vertices is the walk graph of an irreducible necklace with $n$ colours if and only if it satisfies all of the following conditions.
    \begin{enumerate}[\quad a)]        
        \item $G$ is semi-Eulerian but not Eulerian,
        \item for all vertices $v \in V$ we have $\deg(v) \in \{2, 3, 4\}$,
        \item there are no adjacent vertices of degree 2,
        \item the maximum cut of $G$ is at most $\mu(G) \leq n$ .
    \end{enumerate}
\end{restatable}
\begin{proof}
    We first prove that any walk graph of an irreducible necklace fulfils the conditions.
    
    Since any walk graph of some necklace needs to be semi-Eulerian, so is $G$.
    Since the first and last component are of different colours, there are exactly two colours with odd degree, which implies that $G$ is not Eulerian.
    As each colour consists of at most two components the degree of each vertex is at most $4$. Additionally, since the first and last colour are not intervals, the starting and ending vertex have degree $3$ in the walk graph.
    Since there are no two neighbouring intervals in an irreducible necklace, there may not be any adjacent degree $2$ vertices.
    Lastly, the max-cut condition simply captures that the necklace must be $n$-separable.

    It is easy to see that conversely any graph satisfying conditions a)--d) corresponds to an irreducible necklace: one can simply take a Eulerian walk in that graph and place a point of colour $v$ whenever traversing $v$. This results in a necklace where all colours consist of at most two components, no two intervals are neighbouring, the first and the last component are not intervals and the first and last component are of different colour.
\end{proof}

If a graph $G$ is the walk graph for some irreducible necklace, we call $G$ an irreducible walk graph. We will see that all irreducible walk graphs on $n$ vertices are isomorphic. In particular, we claim that the walk graph is isomorphic to the graph $N_n$ defined as follows.
\begin{definition}\label{def:Nn}
    The cycle graph $C_n$ is the cycle on the vertex set $[n]$.
    The graph $P_{n-1}^{odd}$ is the graph with vertex set $[n]$ and edge set
    \[
        E(P_{n-1}^{odd}) = \left\{  \{2i-1, 2i+1\} \mid i \in \left[\lfloor \frac{n-1}{2} \rfloor \right] \right\},
    \]
    that is $P_{n-1}^{odd}$ is the graph on vertex set $[n]$ with a path of length $\lfloor (n-1)/2 \rfloor$ starting at vertex 1 and skipping every other vertex.
    Now the graph $N_{n}$ obtained by taking the union of the edges in $C_n$ and in $P_{n-1}^{odd}$ is called the \emph{irreducible necklace graph} of size $n$.
\end{definition}
See \Cref{fig:irr_walk_graph} for two examples of these graphs.
\begin{figure}
    \centering
    \includegraphics{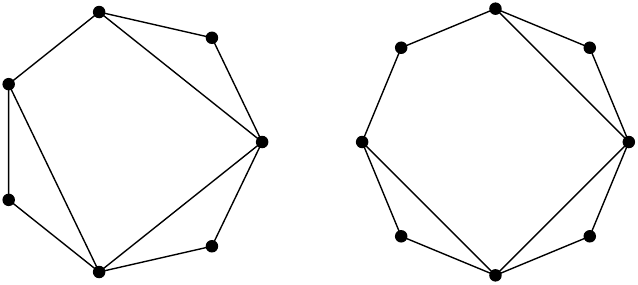}
    \caption{The graphs $N_7$ to the left and $N_8$ to the right.}
    \label{fig:irr_walk_graph}
\end{figure}
Solving $\alpha$-$\overline{\alpha}$-\textsc{Irr-Necklace-Splitting} heavily relies on the following proposition.
\begin{proposition}
    Let $C$ be an irreducible necklace with $n$ colours for some $n \geq 3$.
    The walk graph $G$ of $C$ is isomorphic to $N_n$.
    \label{prop:irreducible-graphs-isomorphic}
\end{proposition}
The proof of this proposition is quite technical and lengthy.
Moreover, the proof itself is not instructive for the further design of the algorithm.
We present the proof in \Cref{app:structure}.

\subsection{Splitting Irreducible Necklaces}
As mentioned above, we will solve $\alpha$-$\overline{\alpha}$-\textsc{Irr-Necklace-Splitting} by formulating it as an edge labelling problem in the walk graph, which will then be formulated as an integer linear program (ILP).
While at first glance a reduction to an ILP might be counterintuitive due to the \textsf{NP}-completeness of ILP, it turns out that in our special case the ILP is tractable in polynomial time.

\subsubsection{The cut labelling problem}
Recall that the edges of the walk graph correspond to intervals between beads of the necklace --- where one of the beads is the last bead of one component and the other is the first bead of a component of some different colour.
Any choice of one component per colour to put a cut point puts some of these intervals on the positive side of the cut, and others on the negative side. In this way, such a choice of components induces a labelling of the edges of the walk graph by ``positive'' and ``negative''. Of course, not every labelling corresponds to a cut. We will now collect some properties of labellings that do correspond to cuts.

For every component of a colour $c$ there are two edges (except if the component appears at the beginning or at the end of the necklace): one edge corresponding to the change of colours when entering the component and one edge when leaving the component.
We call these pairs of edges \emph{traversals}.
For a vertex $v$ define $\trav(v) := \{(e,e') \in E \times E \mid e,e' \text{ is a traversal of $v$}\}$ to be the set of traversals of $v$. The two edges of a traversal are labelled the same if and only if the corresponding component is not chosen to contain a cut point. Thus, if we now consider a vertex corresponding to an interval, it must have exactly one positively labelled incident edge, and one negatively labelled incident edge. Since a colour has exactly one component with a cut point, a bicomponent that is neither the first nor the last colour of the necklace must have either one positive and three negative, or one negative and three positive incident edges.

The edges of the walk graph only contain information about the intervals between components of the necklace. However, we additionally know that if $n$ is even, the interval from $-\infty$ to the first cut point and the interval from the last cut point to $\infty$ must be on the same side of the cut, as there is an even number of cut points.
Similarly for $n$ odd, the two intervals are on opposite sides of the cut. To capture this information in the edge labelling, we slightly modify our walk graph. We add an edge between the vertices corresponding to the first and last component if $n$ is even, or a subdivided edge (with a new degree two vertex) between these two vertices if $n$ is odd. This newly obtained graph is called the \emph{label graph} (see \Cref{fig:irr_colour_graph}). We see that a choice of one component per colour also induces a labelling of these newly added edges. When $n$ is odd, we see that the two edges incident to the newly introduced vertex $v$ must have different labels, just like if $v$ was a regular vertex corresponding to an interval.

\begin{figure}
    \centering
    \includegraphics{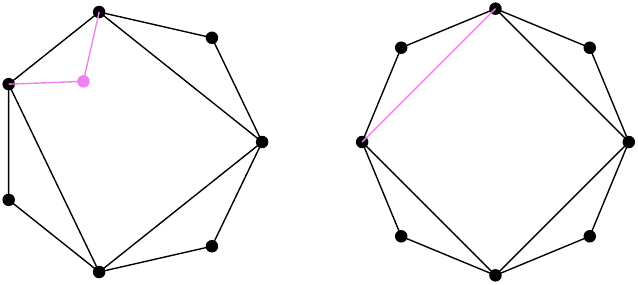}
    \caption{Turning the walk graph to the label graph for $n=7$ (left) and $n=8$ (right). The coloured edges and vertices are added to the walk graph to obtain the label graph.}
    \label{fig:irr_colour_graph}
\end{figure}

So far, all of our properties are invariant under flipping the labelling of all edges. However, by the cut permutation, any choice of one component per colour fixes the positive side of the cut. We can see that for every labelling fulfilling the previous properties, exactly one of the labelling and its inverse are actually the labelling induced by some cut.

We thus have now found a characterisation of the labellings that are induced by choices of one component per colour to cut. However, we are interested only in $\alpha$-cuts. We thus would now like to characterise the labellings that are induced by $\alpha$-cuts.

Clearly, if both edges of a traversal $(e,e')$ are labelled negative, all beads of the corresponding component $c$ of colour $C_v$ must lie on the negative side of the cut. This cut can only be an $\alpha$-cut if $\alpha_c$ is low enough, i.e., $\alpha_v\leq |C_v|-|c|$. Similarly, if both edges were labelled positive, we would have to have $\alpha_v\geq |c|+1$.
It turns out that if an edge labelling of the label graph fulfils this condition for some fixed $\alpha$, we can actually place the cut points in the corresponding components to get an $\alpha$-cut. We thus define the following problem. 

For notational convenience we let $e(v,i,1)$ and $e(v,i,2)$ be the two edges corresponding to the $i$-th traversal of the vertex $v$. That is, $(e(v,i,1),e(v,i,2)) \in \trav(v)$, and this traversal corresponds to the $i$-th component of that colour. In our case $e(v,i,j)$ is only defined for $i \in \{1, 2\}$ for every vertex; for intervals only for $i = 1$. We furthermore define $w_v(i)$ as the size of the $i$-th component of colour $v$.

\begin{definition}\label{def:labelling} We define $\alpha$-\textsc{Cut-Labelling} as the following problem.\\
\textbf{Input:}  The label graph $G = (V, E)$ of some $n$-separable irreducible necklace 
                 $C = \{C_1, \dots, C_n\}$, and a vector $\alpha = (\alpha_1, \dots, \alpha_n)$, $\alpha_i \in \{1, \dots, |C_i|\}$.\\
\textbf{Output:} A subset of the edges $\mathcal{P} \subseteq E$ to be labelled positively (the negatively labelled edges are $\mathcal{N} := E \setminus \mathcal{P}$) such that:
                 \begin{itemize}
                 \item[(1)] $\forall v \in V: |\{e \in \mathcal{P} \mid v \in e\}| \in \{1, 3\}$,
                 \item[(2)] $\forall v \in V, i\in\{1,2\}: \{e(v,i,1),e(v,i,2)\} \subseteq \mathcal{P} \implies \alpha_v \geq w_v(i) + 1$, 
                 \item[(3)] $\forall v \in V,i\in\{1,2\}: \{e(v,i,1),e(v,i,2)\} \subseteq \mathcal{N} \implies \alpha_v \leq |C_v| - w_v(i)$,
                \item[(4)] all edges in $\mathcal{P}$ lie on the positive side of the cut induced by the labelling $(\mathcal{P,N})$.
                \end{itemize}
\end{definition}

A solution of an $\alpha$-\textsc{Cut-Labelling} instance for a given $\alpha$-vector is called an \emph{$\alpha$-labelling}.
The concluding result of this section is the following.
\begin{proposition}
    Let $T_{lab}(n)$ be the time required to solve $\alpha$-\textsc{Cut-Labelling} on the label graph of any irreducible necklace with $n$ colours.
    We can solve $\alpha$-$\overline{\alpha}$-\textsc{Irr-Necklace-Splitting} on any necklace with $n$ colours and $N$ total beads in time $\mathcal{O}(T_{lab}(n, N) + N)$.
    \label{prop:splitting-to-colouring}
\end{proposition}
\begin{proof}
    To solve $\alpha$-$\overline{\alpha}$-\textsc{Necklace-Splitting} on the necklace $C$ we simply solve $\alpha$-\textsc{Cut-Labelling} and $\overline{\alpha}$-\textsc{Cut-Labelling} on its label graph and translate the labellings back to $\alpha$- and $\overline{\alpha}$-cuts. We only consider the $\alpha$-labelling, since the case for $\overline{\alpha}$ works exactly the same way.
    
    We begin by placing a cut point in the first or last bead of each component that is cut according to the $\alpha$-labelling. Thanks to conditions (1) and (4) of an $\alpha$-labelling, this already yields an $\alpha'$-cut for some other $\alpha'$. The cut points can now be moved within their components. Moving the cut point of colour $v$ within its component only changes the number of points on the positive side of colour $v$ and leaves other colours unaffected. Moving the cut point from one side of the component to the other allows us to include any number of points of the component on the positive side, from $1$ point up to all points. Thanks to conditions (2) and (3), each cut point can be moved such that exactly $\alpha_v$ points of the colour are on the positive side. We have thus reached an $\alpha$-cut we can return. Since this movement can be computed for each colour individually, it can be performed in $O(N)$. The label graph can also be computed in $O(N)$, and thus the statement follows.

\end{proof}

Note that for a given $\alpha$, there must be a unique $\alpha$-labelling, since applying the strategy in the proof above to distinct $\alpha$-labellings would yield distinct $\alpha$-cuts, a contradiction to the $\alpha$-Ham Sandwich theorem.

\subsubsection{An ILP formulation}
In this section we model $\alpha$-\textsc{Cut-Labelling} as an Integer Linear Program (ILP).
To do this, we formulate an ILP such that any solution of the ILP corresponds to a labelling fulfilling conditions (1)--(3) in \Cref{def:labelling} and vice versa.
Condition (4) can then be addressed by observing that for a given instance of $\alpha$-\textsc{Cut-Labelling} there are only a constant number of labellings fulfilling conditions (1)--(3): namely the labelling induced by the $\alpha$-cut and the labelling induced by the $\overline{\alpha}$-cut, but with the label of all edges swapped. Instead of merely solving for one solution of the ILP, we will later just enumerate all feasible solutions and check the result against condition (4).

Our ILP formulation will only use binary variables.
For the sake of clarity we use \textbf{bold} face for variables in our ILP.
To model that each edge is either labelled positively or negatively, we introduce two binary variables per edge.
For each edge $e$ in the label graph add two binary variables $\pp_{e}$ and $\nn_{e}$ with the interpretation that $\pp_e = 1 \iff e$ is labelled positive and $\nn_e = 1 \iff e$ is labelled negative.
Since any edge must have exactly one label we add the constraint $\pp_{e} + \nn_{e} = 1$ for every edge $e$.

To model constraints on traversals, we introduce new binary variables $\pp_{v,i}$ and $\nn_{v,i}$ for every vertex $v$ and possible indices of traversals $i$.
The interpretation of these variables are that in the $i$-th traversal of $v$ both the edges are labelled both positive or both negative, respectively. To keep these new variables consistent with the variables for edges, we wish to add the constraints $\pp_{v,i} = \pp_{e(v,i,1)} \cdot \pp_{e(v,i,2)}$ and $\nn_{v,i} = \nn_{e(v,i,1)} \cdot \nn_{e(v,i,2)}$.
Note that these constraints are not linear, but they can be linearised as follows.
\begin{align*}
    \pp_{v,i} &\leq \pp_{e(v,i,1)} \\
    \pp_{v,i} &\leq \pp_{e(v,i,2)} \\
    \pp_{v,i} &\geq \pp_{e(v,i,1)} + \pp_{e(v,i,2)} - 1 \\
     \nn_{v,i} &\leq \nn_{e(v,i,1)} \\
     \nn_{v,i} &\leq \nn_{e(v,i,2)} \\
     \nn_{v,i} &\geq \nn_{e(v,i,1)} + \xx_{e(v,i,2)} - 1
\end{align*}
We also wish to encode whether the edges corresponding to the $i$-th traversal have the \textbf{s}ame label, no matter what this label is.
We thus introduce the variables $\fs_{v,i}$ for each vertex $v$ and its traversals $i$.
They are related to $\pp_{v,i}$ and $ \nn_{v,i}$ as follows.
\[
    \fs_{v,i} = \pp_{v,i} +  \nn_{v,i}
\]
Now we use these variables to encode condition (1) of a cut labelling.
This can be done by the following constraints.
\begin{align*}
    &\forall v \in V, \deg(v) = 4: &\fs_{v,1} + \fs_{v,2} &= 1& & \\
    &\forall v \in V, \deg(v) = 2: &\fs_{v,1} &= 0& &
\end{align*}
Lastly, we need to encode the conditions the $\alpha$-vector poses on the labelling, that is conditions (2) and (3) of a cut labelling.
Note that these conditions are only necessary for vertices with multiple traversals, since for intervals the conditions are always satisfied.

Condition (2) can be implemented using the constraints $\pp_{v,i} \cdot w_v(i) + 1 \leq \alpha_v$.
For condition (3) we need to be more careful since only using $ \nn_{v,i} \cdot w_v(i) \geq \alpha_v$ results in a violated condition if $ \nn_{v,i} = 0$.
We thus need to include the case when $ \nn_{v,i} = 0$ and add a trivial upper bound on $\alpha_v$.
This can be done by using $|C_v|$, the total number of beads for that vertex.
We obtain the following constraints for all vertices $v$ with $\deg(v) =4$.
\begin{align*}
     \nn_{v,1} \cdot w_v(2) + (1- \nn_{v,1})\cdot |C_v| &\geq \alpha_v \\
     \nn_{v,2} \cdot w_v(1) + (1- \nn_{v,2})\cdot |C_v| &\geq \alpha_v \\
    \pp_{v,1} \cdot w_v(1) + 1 &\leq \alpha_v \\
    \pp_{v,2} \cdot w_v(2) + 1 &\leq \alpha_v    
\end{align*}

Note that in general ILPs also optimise an objective function.
However, we are only interested in the satisfying assignments.
To summarise our ILP is given as follows.

\begin{definition}$\alpha$-\textsc{Cut-Labelling-ILP}

Let $G = (V,E)$ be a label graph of some $n$-separable necklace $C = \{C_1,\dots, C_n\}$ with colours consisting of at most two components. 
The $\alpha$-\textsc{Cut-Labelling-ILP} is the ILP given by the feasibility of the following constraints.
\begin{align*}
    \forall e \in E: &&\pp_{e}, \nn_{e} &\in \{0,1\} \\
    \forall v \in V,i\in[deg(v)/2]: &&\pp_{v,i},  \nn_{v,i}, \fs_{v,i} &\in \{0,1\} \\
    \forall e \in E: &&\pp_{e} + \nn_{e} &= 1 \\
    \forall v \in V,i\in[deg(v)/2]: &&  \pp_{v,i} &\leq \pp_{e(v,i,1)} \\
                     && \pp_{v,i} &\leq \pp_{e(v,i,2)} \\
                     && \pp_{e(v,i,1)} + \pp_{e(v,i,2)} - 1 &\leq \pp_{v,i}  \\
                     &&  \nn_{v,i} &\leq \nn_{e(v,i,1)} \\
                     &&  \nn_{v,i} &\leq \nn_{e(v,i,2)} \\
                     && \nn_{e(v,i,1)} + \nn_{e(v,i,2)} - 1 &\leq  \nn_{v,i} \\  
                     && \fs_{v,i} &= \pp_{v,i} +  \nn_{v,i} \\
    \forall v \in V, \deg(v) =2: &&\fs_{v,1} &= 0 \\
    \forall v \in V, \deg(v) =4: && \fs_{v,1} + \fs_{v,2} &= 1 \\
             && \nn_{v,1} \cdot w_v(2) + (1- \nn_{v,1})\cdot |C_v| &\geq \alpha_v \\
             && \nn_{v,2} \cdot w_v(1) + (1- \nn_{v,2})\cdot |C_v| &\geq \alpha_v \\
             &&\pp_{v,1} \cdot w_v(1) + 1 &\leq \alpha_v \\
             &&\pp_{v,2} \cdot w_v(2) + 1 &\leq \alpha_v 
\end{align*}    
\end{definition}

By construction, $\alpha$-\textsc{Cut-Labelling-ILP} is equivalent to conditions (1)--(3) of $\alpha$-\textsc{Cut-Labelling} in the sense that any solution of the former can be turned to a labelling fulfilling the conditions and vice versa, by simply taking the labelling given by $\pp_e$ and $\nn_e$ and vice versa.

\subsubsection{Solving cut labelling for irreducible necklaces}
In this section we show that $\alpha$-\textsc{Cut-Labelling-ILP} is solvable in polynomial time.
To do that we leverage the structure of the ILP given by the fact that the underlying necklace is irreducible. We wish to use the following FPT algorithm by Jansen and Kratsch.

\begin{theorem}[\cite{ILPtreewidth}]\label{thm:ILPtreewidth}
    For an ILP instance $\mathcal{I}$ where each variable is in the domain $\mathcal{D}$, feasibility can be decided in time $\mathcal{O}(|\mathcal{D}|^{\mathcal{O}(\mathsf{tw}(P(\mathcal{I})))}\cdot |\mathcal{I}|)$, where $\mathsf{tw}(P(\mathcal{I}))$ denotes the treewidth of the primal graph of the instance.
    Moreover, if the ILP only has a constant number of feasible solutions, they can be enumerated in the same time bound.
    \label{thm:fpt-ilp}
\end{theorem}

To understand this theorem we need the following two definitions.
\begin{definition}[Primal graph]
    For an ILP $\mathcal{I}$ given by $\min c^\top \xx$ subject to $A \xx \leq b$, the \emph{primal graph} $P(\mathcal{I})$ is the graph having a vertex for each variable $\xx_i$ in $\mathcal{I}$ and an edge between two variables $\xx_i$ and $\xx_j$ if and only if there is a constraint $A_k$ such that $A_{k,i} \neq 0$ and $A_{k,j} \neq 0$.
    That is, there is an edge if and only if $\xx_i$ and $\xx_j$ occur together in a constraint.
\end{definition}
\begin{definition}[Treewidth]
    Let $G = (V,E)$ be a graph. A \emph{tree decomposition} $(\mathcal{T} = (\mathcal{V}, A), B: \mathcal{V} \to 2^{V})$ consists of a tree $\mathcal{T}$ and a mapping $B$ mapping vertices of $\mathcal{T}$ to subsets of vertices of $G$, commonly referred to as bags, such that the following three conditions hold.
    \begin{enumerate}
        \item $\bigcup_{x \in \mathcal{V}} B(x) = V$, 
        \item $\forall e = \{u,v\} \in E: \exists x \in \mathcal{V} : \{u,v\} \subseteq B(x)$,
        \item $\forall v \in V: \mathcal{T}[\{x \in \mathcal{V} \mid v \in B(x)\}]$ is connected.
    \end{enumerate}
    The \emph{width} of a tree decomposition, $\mathsf{wd}(\mathcal{T}, B)$, is then given by the size of a largest bag minus 1.
    That is, $\mathsf{wd}(\mathcal{T}, B) := \max \{|B(x)| \mid x \in \mathcal{V}\} - 1$.
    The \emph{treewidth} of $G$, denoted by $\mathsf{tw}(G)$, is then given by the minimal width among all tree decompositions of $G$, formally 
    \[  
        \mathsf{tw}(G) := \min_{(\mathcal{T},B) \text{ tree decomposition of $G$} } \mathsf{wd}(\mathcal{T}, B).
    \]
    \label{def:treewidth}
\end{definition}
Treewidth is a very well-studied graph parameter. We now wish to show that the primal graph of the $\alpha$-\textsc{Cut-Labelling-ILP} has a constant treewidth. Since the ILP is boolean, the domain is $\mathcal{D} = \{0,1\}$, and thus from this \Cref{thm:ILPtreewidth} gives us polynomial runtime for enumerating all solutions of  $\alpha$-\textsc{Cut-Labelling-ILP}, and thus for solving $\alpha$-\textsc{Cut-Labelling}.

Our approach for bounding the treewidth will be of a hierarchical nature. We assume that this approach is not novel, but we were unable to find a fitting reference.
\begin{definition}
    Let $G = (V, E)$ be a graph. An \emph{enhanced graph} of $G$ is a graph $\mathcal{G} = (\mathcal{X}, \mathcal{Y})$, with vertices $\mathcal{X} = \{X_1, \dots, X_k\}$ with $X_i \subseteq V$, where the $X_i$ are called hypervertices, such that $\bigcup_{X_i \in \mathcal{X}} X_i = V$.
    The edges of $\mathcal{G}$ are given by 
    \[
        \mathcal{Y} = \left\{\{X_i, X_j\} \subseteq \mathcal{X} \mid X_i \cap X_j \neq \emptyset \text{ or } \exists \{u,v\} \in E: u \in X_i \text{ and } v \in X_j\right\}.
    \]
\end{definition}
An enhanced graph $G$ can be viewed as a graph capturing the structure of $G$.
The vertices of the enhanced graph are subgraphs of $G$.
Hypervertices are connected if they share an element or there is an edge in $G$ connecting an element from one hypervertex and an element from the other.
The treewidth of an enhanced graph naturally gives an upper bound on the treewidth of the graph itself, as seen in the following lemma.
\begin{lemma}
    Consider a graph $G = (V, E)$ with an enhanced graph $\mathcal{G} = (\mathcal{X}, \mathcal{Y})$ of $G$.
    Let $w := \max \{|X_i| \mid X_i \in \mathcal{X}\}$ be the size of a largest hypervertex in $\mathcal{G}$.
    Then the treewidth of $G$ is at most $\tw(G) \leq w \cdot (\tw(\mathcal{G})+1) - 1$.
    \label{lem:treewidth-enhancement}
\end{lemma}
We include the proof of \Cref{lem:treewidth-enhancement} in \Cref{app:enhancedgraph}.

Now we are ready to show that the treewidth of the primal graph of the ILP for irreducible necklaces is bounded.
We show the following proposition.
\begin{proposition}
    Let $C = \{C_1, \dots, C_n\}$ be an irreducible necklace with $n$ colours.
    Then the primal graph of the $\alpha$-\textsc{Cut-Labelling-ILP} of that necklace has treewidth at most 55.
    \label{prop:treewidth-bounded}
\end{proposition}
\begin{proof}
    We proceed in the following steps.
    First we show how the primal graph of that ILP looks like.
    Then we construct an enhanced graph $\mathcal{G}_1$ of that primal graph.
    In a next step we analyze $\mathcal{G}_1$ and give an enhanced graph $\mathcal{G}_2$ of $\mathcal{G}_1$, which turns out to be isomorphic to the label graph.
    Since we understand the label graph we can bound its treewidth and we can bound the treewidth of the primal graph by applying \Cref{lem:treewidth-enhancement}.

    First of all recall that \Cref{prop:irreducible-graphs-isomorphic} characterises the walk graphs of irreducible necklaces.
    This implies that the label graph of an irreducible necklace with $n$ colours is always isomorphic to $C_k + C^{odd}_k$, for $k = 2\lceil n / 2 \rceil$.
    Here $C_k$ is a cycle graph with vertices $[k]$ and $C^{odd}_k$ is the cycle graph with the odd vertices in $[k]$.

    Consider the two traversals of a vertex $v$ with degree 4 as depicted in \Cref{fig:ilp_structure_traversal.pdf}. In the following we will use the vertex and edge names as in \Cref{fig:ilp_structure_traversal.pdf}.
    
    \begin{figure}[htb]
        \centering
        \includegraphics{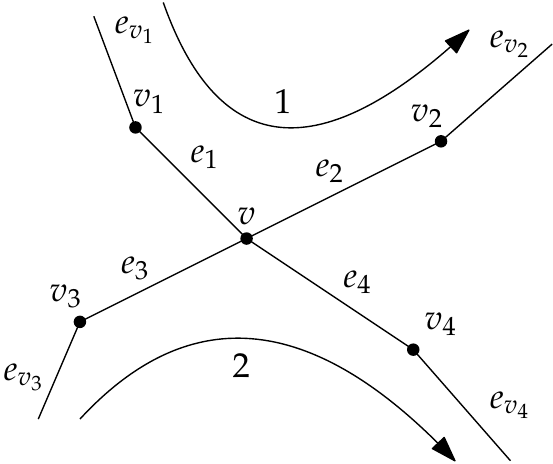}
        \caption{Two traversals in the label graph for a vertex $v$. The arrows indicate the order in which the edges are traversed: in the first traversal $e_1$ and $e_2$ and in the second $e_3$ and $e_4$.}
        \label{fig:ilp_structure_traversal.pdf}
    \end{figure}
    
    Recall that in the primal graph we have a vertex for each variable and an edge between variables if there is a constraint using both these variables. We illustrate all the variables corresponding to the two traversals of $v$, as well as the subgraph of the primal graph induced by these variables in \Cref{fig:ilp_structure_primal.pdf}.
    
    In particular, the primal graph captures the same cyclic structure as the Euler tour in the label graph defined by the necklace, where additionally there are chords between the two occurrences of the same vertex.

    \begin{figure}[htb]
        \centering
        \includegraphics{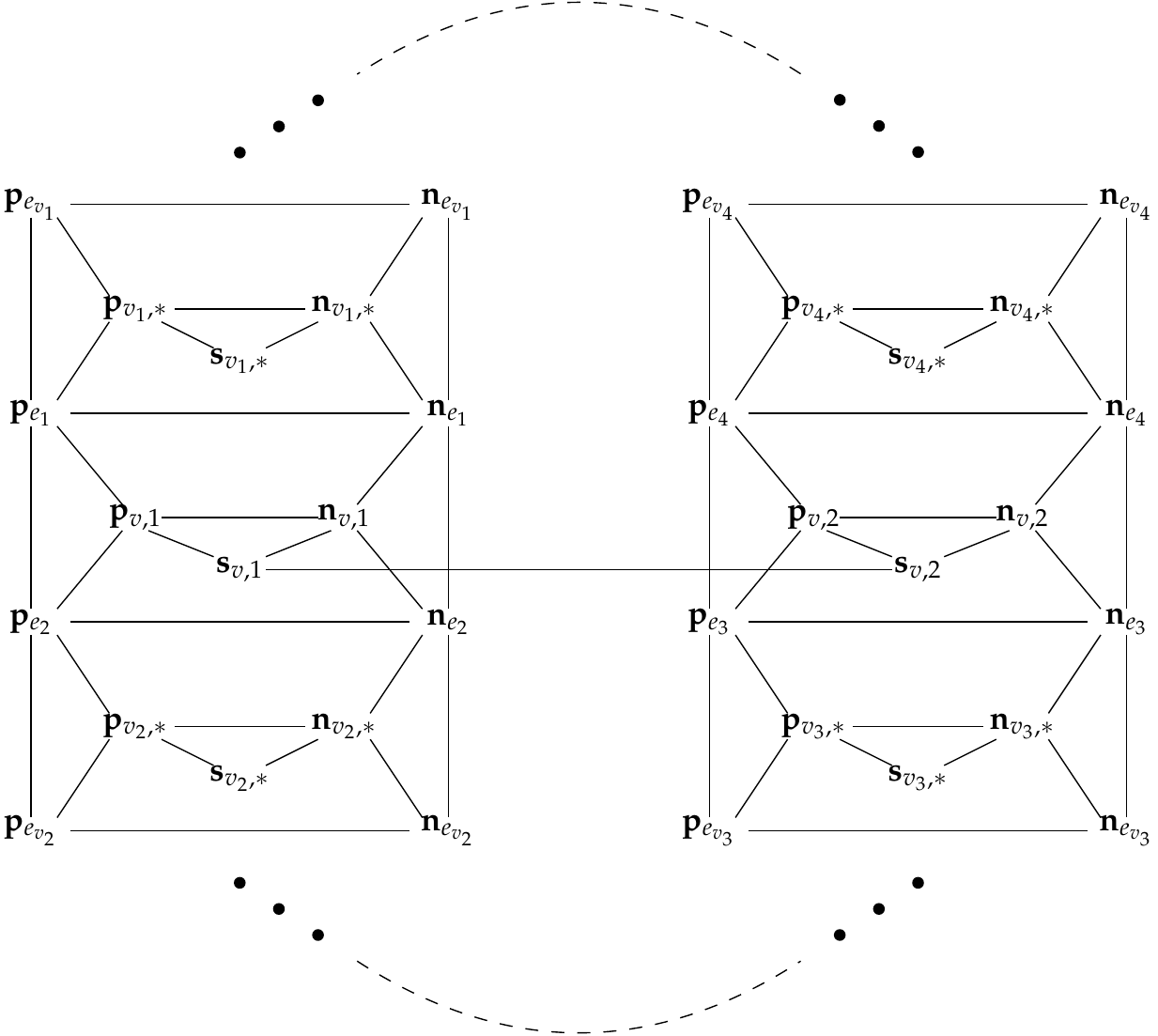}
        \caption{The structure of the primal graph of the ILP. The variables are named according to \Cref{fig:ilp_structure_traversal.pdf}. The *'s indicate that the traversal index for these variables is not relevant for the local picture. The dots and dashed line indicate that the graph continues in a cyclic structure. This cyclic structure corresponds to the Euler tour in the label graph corresponding to the necklace.}
        \label{fig:ilp_structure_primal.pdf}
    \end{figure}

    As mentioned before, we now construct an enhanced graph of the primal graph.
    To do this, we need to define the set of hypervertices as the set of variables occurring in a traversal of a vertex.
    That is if $e=e(v,i,1),e'=e(v,i,2)$ are the edges of the $i$-th traversal of a vertex $v$, we define $X_{v,i} := \{\pp_{e},\nn_{e},\pp_{e'},\nn_{e'}, \pp_{v,i},  \nn_{v,i}, \fs_{v,i}\}$.
    Therefore, define the enhanced graph $\mathcal{G}_1$ of the primal graph to be the enhanced graph defined by $X_{v,i}$ for all vertices $v$ and traversals $i$.
    Note that this is indeed an enhanced graph as any variable occurs in at least one traversal of some vertex.

    Now observe that $\mathcal{G}_1$ is a long cycle corresponding to the Euler tour of the necklace, with the $i$-th traversal of a vertex $v$ being represented by a vertex $X_{v,i}$ in $\mathcal{G}_1$, and these vertices  are adjacent as they appear in the Euler tour.
    Moreover, for a vertex $v$ with two traversals, the corresponding vertices $X_{v,1}$ and $X_{v,2}$ are connected by a chord.
    This edge stems from the edge between $\fs_{v,1}$ and $\fs_{v,2}$ in the primal graph.
    The structure of $\mathcal{G}_1$ is depicted in \Cref{fig:ilp_structure_hypercycle1.pdf}.
    
    \begin{figure}[htb]
        \centering
        \includegraphics{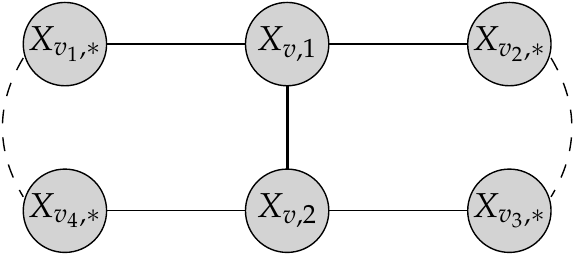}
        \caption{The structure of $\mathcal{G}_1$, where the labels correspond to the labels from \Cref{fig:ilp_structure_traversal.pdf} and \Cref{fig:ilp_structure_primal.pdf}. The vertical edge is the one stemming from the edge between $\fs_{v,1}$ and $\fs_{v,2}$. The dashed line indicates that there is a path between two vertices.}
        \label{fig:ilp_structure_hypercycle1.pdf}
    \end{figure}

    \begin{figure}[htb]
        \centering
        \includegraphics{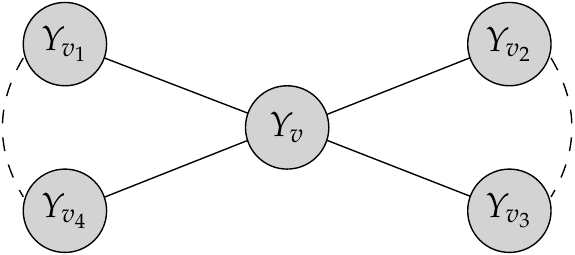}
        \caption{The structure of $\mathcal{G}_2$. Note that the only difference to $\mathcal{G}_1$ is that we grouped the edges stemming from $\fs_{v,1}$ and $\fs_{v,2}$ into a single hypervertex.
        Then $\mathcal{G}_2$ is isomorphic to the label graph as in \Cref{fig:ilp_structure_traversal.pdf}.}
        \label{fig:ilp_structure_hypercycle2.pdf}
    \end{figure}

    We can now define an enhanced graph $\mathcal{G}_2$ of $\mathcal{G}_1$ by simply grouping the vertices $X_{v,1}$ and $X_{v,2}$ together, for all vertices $v$.
    Hence, in $\mathcal{G}_2$ a vertex $Y_v$ has edges to all vertices $Y_{v'}$ where $v'$ is adjacent to $v$ in the label graph.
    This implies that $\mathcal{G}_2$ is isomorphic to the label graph.
    This is depicted in \Cref{fig:ilp_structure_hypercycle2.pdf}.
   
    Now observe that the label graph -- and thus $\mathcal{G}_2$ -- has treewidth at most $3$.
    This can be seen by making a bag for each interval together with its two adjacent vertices.
    The tree is then given by a path obtained by walking along the cycle of the label graph.
    Since the vertex at the start of the cycle will appear in the first and the last bag, we just add it to every bag to get the third property of a tree decomposition.
    Hence, every bag contains at most four vertices, yielding $\tw(\mathcal{G}_2) \leq 3$.
    By construction the hypervertices of $\mathcal{G}_2$ contain at most two vertices of $\mathcal{G}_1$.
    Using \Cref{lem:treewidth-enhancement} we therefore have $\tw(\mathcal{G}_1) \leq 2\cdot4 -1 = 7$.
    Furthermore, since the hypervertices of $\mathcal{G}_1$ consist of exactly seven vertices of the primal graph, we conclude that the treewidth of the primal graph is at most $7\cdot 8 - 1 = 55$.
\end{proof}

We are now ready to prove the following.
\begin{proposition}
    $\alpha$-$\overline{\alpha}$-\textsc{Irr-Necklace-Splitting} can be solved in $\mathcal{O}(n \cdot N)$ time.
    \label{prop:solving-irreducible}
\end{proposition}
\begin{proof}
    \Cref{thm:fpt-ilp} implies an $\mathcal{O}(n^2)$ algorithm for $\alpha$-\textsc{Cut-Labelling-ILP}, using that $|\mathcal{D}|= |\{0,1\}| \in \mathcal{O}(1)$ and $\tw(P(\mathcal{I})) \in \mathcal{O}(1)$ by \Cref{prop:treewidth-bounded}.
    The size of the ILP instance $|\mathcal{I}|$ is given by the number of entries of the matrix defining the instance $\mathcal{I}$.
    Thus we have $|\mathcal{I}| \in \mathcal{O}(n \cdot n)$, as we have a constant number of variables and constraints per colour.
    As argued above, the ILP only has a constant number of feasible solutions.
    Hence, we can enumerate these efficiently, and determine which of these satisfy condition (4) of $\alpha$-\textsc{Cut-Labelling}.
    This check can be implemented in $\mathcal{O}(N)$.
    This shows that $\alpha$-\textsc{Cut-Labelling} can be solved in time $\mathcal{O}(n^2+N)$.
    Therefore, by \Cref{prop:splitting-to-colouring} $\alpha$-$\overline{\alpha}$-\textsc{Irr-Necklace-Splitting} can be solved in $\mathcal{O}(n^2 + N) \in \mathcal{O}(n\cdot N)$ time, using that $n \in \mathcal{O}(N)$.
\end{proof}

Together with \Cref{prop:reduce-to-two-components,prop:reduce-to-irreducible}, we can conclude that $\alpha$-$\overline{\alpha}$-\textsc{Necklace-Splitting}, and thus $\alpha$-\textsc{Necklace-Splitting}, can be solved in $\mathcal{O}(n \cdot N)$ time. We have thus proven \Cref{thm:mainAlgo}.

\section{\NP-Completeness of the Decision Problem}
In this section we prove \Cref{thm:NPhardness}, i.e., we show \textsf{NP}-completeness of deciding whether an arbitrary necklace has an $\alpha$-cut. Since an $\alpha$-cut can be used as a yes-certificate and verified in polynomial time, this problem is clearly in \NP.

Instead of showing \NP-hardness of this problem directly, we instead show \NP-hardness of the following problem.

\begin{definition}$\alpha/\overline{\alpha}$-\textsc{Necklace-Deciding}

\begin{tabular}{ll}
\textbf{Input:}  & A necklace $C = \{C_1, \dots, C_n\}$ and $\alpha = (\alpha_1, \dots, \alpha_n)$, $\alpha_i \in \{1, \dots,|C_i|\}$\\
\textbf{Decide:} & Does $C$ have a cut that is either an $\alpha$-cut or an $\overline{\alpha}$-cut?
\end{tabular}
\end{definition}

$\alpha/\overline{\alpha}$-\textsc{Necklace-Deciding} can be seen as the problem of deciding whether there exist cut points that split the necklace into two sides, one of which has size $\alpha$, but not necessarily the positive side of the cut defined via the cut parity. Clearly, $\alpha/\overline{\alpha}$-\textsc{Necklace-Deciding} can be solved by testing individually whether $C$ has an $\alpha$-cut, and whether it has an $\overline{\alpha}$-cut. Thus, the following proposition immediately implies \Cref{thm:NPhardness}.

\begin{proposition}
    $\alpha/\overline{\alpha}$-\textsc{Necklace-Deciding} is \textsf{NP}-hard.
\end{proposition}
\begin{proof}
    We reduce from \textsc{e3-Sat}.
    Let $\Phi = C_1 \land \dots \land C_m$ be an instance of \textsc{e3-Sat}.
    Each clause $C_i$ consists of exactly three variables.
    We now construct a necklace $\zeta$ and a vector $\alpha$ such that $\zeta$ has an $\alpha$- or $\overline{\alpha}$-cut if and only if $\Phi$ is satisfiable. As discussed above, we can instead show that $\Phi$ is satisfiable if and only if there exist cut points (one per colour) such that \emph{some} chosen side contains the desired number of points $\alpha$.

    To construct $\zeta$, we use the following types of beads.
    The two types of beads $P$ and $N$ are used to enforce that certain parts of the necklace are on the positive or negative side (we will set $\alpha(P)=|P|$ and $\alpha(N)=1$).
    For every variable $x$ we use the types $x_0^A$ and $x_0^B$ to construct a gadget which encodes the truth value of $x$.
    For every clause $C_i$ such that $x \in C_i$ we additionally use the types $x_i^A$ and $x_i^B$ to read out the value of $x$ at the clause $C_i$.
    To transfer the value of $x$, we use yet another type of bead $x_T$.
    For clauses $C_i$ we only need one more type of bead, which we simply name $C_i$.
    
    We also need multiple types of beads as \emph{separator beads}. A separator bead only occurs once, and is used to enforce that there is a cut at its location. We denote a separator bead type by $S_*$ where the subscript indicates the part of the necklace the separator bead belongs to.

    Finally, we need two types of beads $a,b$ that are only used to enforce the positive side of the cut.
    The necklace starts with the string $a\,b\,a$ and we set $\alpha$ such that both beads of $a$ must be on the positive side.
    Since there must be exactly one cut through one bead of $a$ and another cut through the single occurrence of $b$, the only way of having both $a$ and the $b$ on the positive side is to cut the second occurrence of $a$, and to put the unbounded interval from $-\infty$ to the first occurrence of $a$ on the positive side of the cut, or to cut through the first occurrence of $a$, and still putting the unbounded interval from $-\infty$ to the first bead of $a$ on the positive side, which will then also put the second bead of $a$ on the positive side.

    The rest of the necklace consists of two main parts.
    The first part is the encoding part where we encode the assignment of variables and the satisfiability of each clause.
    In the second part we enforce the position of cuts for some types of beads.
    In the following we describe both parts, starting with the encoding part.

    The encoding part first contains the following string for each variable $x$
    \[
        P\, x_0^A\, \underbrace{x_T\dots x_T}_{k\text{ times}}\, x_0^B\, P\, x_0^A\, x_0^B\, P
    \]
    where $k$ is the number of occurrences of the variable $x$, that is, the number of clauses $C_i$ such that $x \in C_i$ (either positively or negatively).
    Then, the encoding part contains the following string for each clause $C_i$ and each variable $x \in C_i$.
    \begin{align*}  
        P\, x_i^A\, &&C_i&&\, &x_i^B\, P\, x_i^A\, x_T\,&& && &x_i^B\, P & &\text{ if $x$ appears as a positive literal in $C_i$, }\\
        P\, x_i^A\, && && &x_i^B\, P\, x_i^A\, x_T\,&&C_i&&\, &x_i^B\, P& &\text{ if $x$ appears as a negative literal in $C_i$.}
    \end{align*}
    Note that the only difference between these two strings is the placement of the bead of type~$C_i$.

    We prove that under the following assumptions the encoding part properly encodes true assignments to $\Phi$. These assumptions will be enforced by the $\alpha$-vector and the second part of the necklace.
    We first assume that there is no cut in any bead of types $P, x_T, C_i$ within the encoding part. Furthermore, we assume that a cut is correct if and only if the positive side of the cut \emph{within the encoding part} contains all beads of type $P$, half of each $x_T$ (note that $x_T$ occurs an even number of times), at most two beads of type $C_i$, and both beads of types $x_i^A$ and $x_i^B$ (for all $i$ including 0). Recall that a bead at which we cut the necklace is counted towards the positive side.
    
    The only way to make the $x_i^A$ and $x_i^B$ types have both beads on the positive side is by either cutting through the first $x_i^A$ and first $x_i^B$, or by cutting through the second $x_i^A$ and second $x_i^B$.
    The first case will encode the assignment $x = \text{``true''}$ and the latter encodes $x = \text{``false''}$.

    The cuts for $i=0$ corresponding to a true assignment will move the $k$ beads of type $x_T$ between $x_0^A$ and $x_0^B$ to the negative side.
    Since we need half of the beads of type $x_T$ on the positive side, this implies that for the cuts in $x_i^A$ and $x_i^B$ for $i\neq 0$, the bead of type $x_T$ must be on the positive side. Hence, the cuts in $x_i^A$ and $x_i^B$ must encode the same assignment as the cuts in $x_0^A$ and $x_0^B$.

    We see that the bead of type $C_i$ between $x_i^A$ and $x_i^B$ is on the negative side of the cut if and only if $C_i$ is satisfied by the variable $x$. We thus see that at most two of the beads of type $C_i$ are on the positive side if and only if $C_i$ is fulfilled by one of its literals.
    Hence, we can conclude that under the assumptions listed above, the encoding part correctly encodes correct assignments to $\Phi$.
    
    In the second part of the necklace we will now first enforce that the cut points of type $P, x_T,$ and $C_i$ will lie in this second part.

    We add the following three strings. For each clause $C_i$ we add the string
    \[  
        P\, C_i\, C_i\, C_i\, S_i\, P,
    \]
    where $S_i$ is a new separator bead. Still assuming that $P$ is not getting cut, this enforces that one of the three beads of type $C_i$ is cut. 

    Then, for each variable $x$ we add the string
    \[  
    \underbrace{P\,x_T\,\dots\,P\,x_T}_{k \text{ times}}\underbrace{N\,x_T\,\dots\,N\,x_T}_{k \text{ times}}\,N\,S_x,
    \]
    where again $k$ is the number of clauses containing $x$, and $S_x$ is a new separator bead. Assuming $P$ and $N$ to not be cut, this enforces the $k$-th occurrence of $x_T$ to be cut, putting $k$ of the beads of type $x_T$ in this substring on the positive side, and $k$ on the negative side.

    Finally, we add the string
        \[
        P\, N.
    \]
    Since the number of types of beads is even ($P$ comes paired with $N$, $x_i^A$ comes paired with $x_i^B$, $C_i$ is paired with $S_i$ and $x_T$ is paired with $S_x$), we can see that this last bead of type $N$ must be cut: If it was not cut, it would be on the positive side of the cut, since both unbounded intervals must belong to the positive side. However, the cut point of type $N$ is another bead that would be counted to the positive side, violating $\alpha(N)=1$. If however this last bead of type $N$ is cut, also the last bead of type $P$ must be cut by a symmetric argument.

    We thus see that the second part of the string enforces the cuts we assumed in the encoding part. Furthermore, since the second part has exactly half of  the beads of type $x_T$ and one up to three beads of type $C_i$ on the positive side, the following $\alpha$-vector exactly gives us all of the assumptions on the number of beads on the positive side we made in the encoding part.
    \begin{align*}
        &\alpha(a) = 2,\ \alpha(b) = 1, \\
        &\alpha(P) = |P|,\ \alpha(N) = 1,\\
        &\alpha(x_i^A) = 2,\ \alpha(x_i^B) = 2 &\text{ for all variables $x$ and all $i$ including 0,}\\
        &\alpha(x_T) = \frac{|x_T|}{2} & \text{ for all variables $x_T$, note that $|x_T|$ is even,}\\
        &\alpha(C_i) = 3 & \text{ for all clauses $C_i$,}\\
        &\alpha(S_*) = 1 & \text{ for any separator $S_*$.}
    \end{align*}
    We have argued before that a cut with the correct number of points $\alpha$ on the positive side corresponds to a satisfying assignment of $\Phi$. On the flip side, it is clear that a satisfying assignment can be turned into a cut by cutting the $x_i^A,x_i^B$ at the corresponding places and the $C_i$ at the correct of the three consecutive occurrences, depending on how many literals in $C_i$ are true.

    The necklace $\zeta$ can be built in polynomial time in $m$, and we thus get that $\alpha/\overline{\alpha}$-\textsc{Necklace-Deciding} is \textsf{NP}-hard.
\end{proof} 

We include an example of this reduction, as well as of the correspondence between cuts and assignments in \Cref{app:example}.

\section{Conclusion and Further Directions}
We have settled the complexity of the promise search problem of \UnfairSplitting on $n$-separable necklaces by showing it to be polynomial-time solvable. Together with the result of Borzechowski, Schnider and Weber~\cite{nseparableNecklaces} showing that $n$-separability is polynomial-time checkable, this also shows the total search problem corresponding to \UnfairSplitting to be in \textsf{FP}. We have contrasted this by showing that the decision problem for existence of an $\alpha$-cut is \NP-complete. It is conceivable that an FPT algorithm using the parameter $\ell$ as in~\cite{nseparableNecklaces} for \UnfairSplitting may exist. Such an algorithm would decide whether a given $(n-1+\ell)$-separable necklace has an $\alpha$-cut, and find one if it does, in time $f(\ell)\cdot \text{poly}(n+N)$. To adapt our algorithm to yield such an FPT algorithm, one would need to analyse the walk graphs of irreducible $(n-1+\ell)$-separable necklaces, and bound their treewidth in terms of $\ell$.

It remains open whether $\alpha$-Ham Sandwich is \UEOPL-complete. To prove \UEOPL-hardness, we now know that we need to find new problems to reduce from, and \UnfairSplitting cannot be used.

\bibliography{literature}
\newpage

\appendix

\section{Proving \texorpdfstring{\Cref{prop:irreducible-graphs-isomorphic}}{Proposition \ref{prop:irreducible-graphs-isomorphic}}}\label{app:structure}
In this section we prove \Cref{prop:irreducible-graphs-isomorphic}.
Recall the statement of said proposition: It claims that for any irreducible necklace $C$ with $n$-colours, the walk graph of $C$ is always isomorphic to the irreducible necklace graph $N_n$, defined as $N_n = C_n + P_{n-1}^{odd}$ in \Cref{def:Nn}.
Recall that $C_n$ is the cycle graph with vertices $[n]$ and $P_{n-1}^{odd}$ is the graph on vertex set $[n]$ with a path of length $\lfloor (n-1)/2 \rfloor$ on the odd vertices.
Also recall the conditions for the walk graphs of irreducible necklaces that we already derived in \Cref{lem:irreducibleconditions}.

\irreducibleCond*
\vspace{0.2cm}

To show \Cref{prop:irreducible-graphs-isomorphic} we distinguish the cases when $n$ is even and when $n$ is odd.
We show the result for even $n$ directly and later we will reduce the case for odd $n$ to the even case.
It is easy to verify that for $n \leq 4$ \Cref{prop:irreducible-graphs-isomorphic} holds, as in this case the only irreducible necklaces (up to permutations) are the following, as verified computationally.
\[
\includegraphics[width=0.9\linewidth]{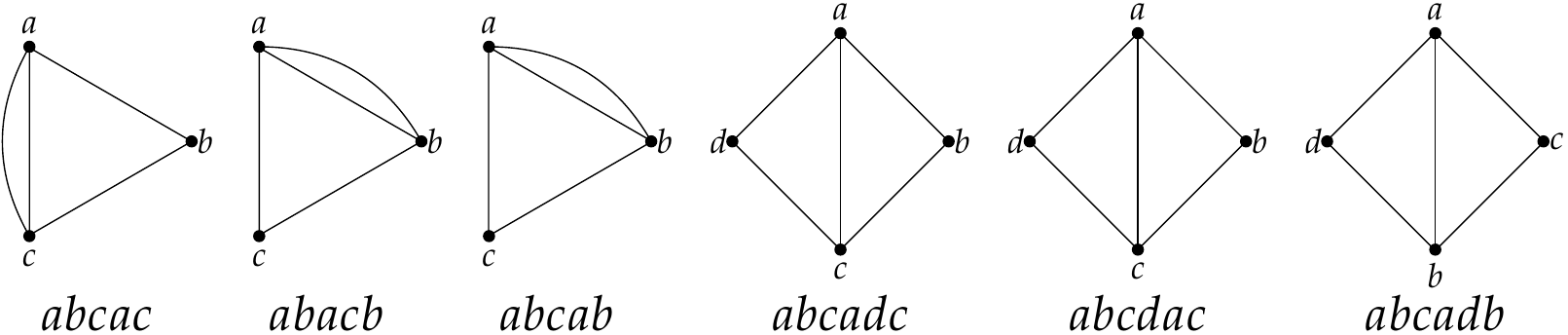}
\]
All of these necklaces have a walk graph isomorphic to $N_n$.
We can therefore assume $n \geq 5$ throughout the remainder of this section.

\subsection{Even \texorpdfstring{$n$}{n}}
We first focus on the case that $n$ is even.
Throughout, let $G$ be an irreducible walk graph, that is, a walk graph for some irreducible necklace $C$ with $n$ colours.
To prove \Cref{prop:irreducible-graphs-isomorphic}, we proceed as follows.
The first and main step is to show that in $G$ we have the correct number of intervals\footnote{In the following we interchangeably use the terms intervals and degree 2 vertices.
Similarly, we use bicomponent to refer to vertices of degree at least 3.}.
Having that, we will show that in $G$ we find the same local structures as in $N_n$.
In particular, we will show that bicomponents are adjacent to exactly two intervals and that the adjacent bicomponents of an interval are adjacent.
Then, we can prove the graph isomorphism by giving an explicit construction of $N_n$ from $G$.
The main tool in these steps is to derive a contradiction by constructing a cut in $G$ that implies $\mu(G) > n$.

We start by analysing the number of intervals in $G$.
Let $I$ be the set of intervals in $G$ and denote its size with $k = |I|$.
First notice that condition c) of irreducible walk graphs states that there are no adjacent intervals.
Moreover, by condition b) all intervals have degree 2, and thus the cut given by $I$ has size $2k$.
By condition d) we have $2k \leq \mu(G) \leq n$, which implies that $k \leq n/2$.
Additionally, \Cref{lem:num-intervals} directly implies that $k \geq \frac{n-3}{2}$.
This shows that for any $n$ (even or odd), we have that $k \in \{\lfloor n / 2 \rfloor - 1, \lfloor n/2 \rfloor\}$.
We summarise this result in the following.
\begin{corollary}
    Let $G$ be an irreducible walk graph with $n$ vertices for some $n \geq 5$.
    Then the number of vertices of degree 2, denoted by $k$, is $k \in \{\lfloor n / 2 \rfloor - 1, \lfloor n/2 \rfloor\}$.
    \label{lem:num-intervals-irreducible-almost}
\end{corollary}

Observe that the irreducible necklace graph $N_n$ has exactly $\lfloor n / 2 \rfloor$ intervals.
Thus, we need to show that $k = \lfloor n / 2 \rfloor$, which is $k= n/2$ for even $n$.
We will do this by contradiction, that is, we show that $k = n/2 -1$ implies that $\mu(G) > n$, contradicting condition d) of irreducible walk graphs.
To do so, we first need some notation.

Let $E_I$ be the set of edges incident to intervals.
For any $e \in E_I$ we label the edge with the bicomponent incident to e.
That is, for an edge $e = \{u,v\} \in E_I$ with $\deg(u) = 2$ and $\deg(v) \geq 3$, we label $e$ with $l(e) := v$.
This label is well-defined, since all edges in $E_I$ are incident to exactly one bicomponent thanks to property c).

Now let $L \subseteq V(G)$ be the multiset of labels of the edges in $E_I$, that is, each label is in $L$ once for every edge with that label.
Formally,
\[
    L = \{v^{m(v)} \mid v \in V(G),\ \deg(v) \geq 3 \}
\]
where $m(v) := |\{e \in E_I \mid l(e) = v\}|$ is the multiplicity of the label $v$.
For an example of this notation, see \Cref{fig:multiplicity_example}.
\begin{figure}
    \centering
    \includegraphics{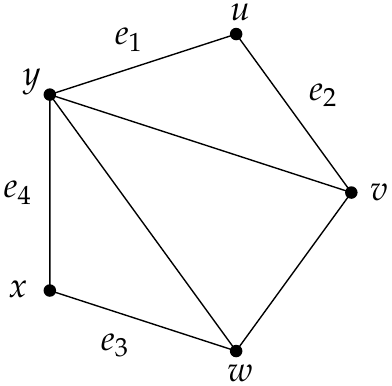}
    \caption{This walk graph contains two intervals, $u$ and $x$. The edges incident to these intervals are $E_I = \{e_1, e_2, e_3, e_4\}$.
    The labels of these edges are given by $l(e_1) = y$, $l(e_2) = v$, $l(e_3) = w$, $l(e_4) = y$.
    Hence, in $L$ the vertices $v$ and $w$ have multiplicity 1 and $y$ has multiplicity 2.}
    \label{fig:multiplicity_example}
\end{figure}

Assuming $k = n/2 -1$, we can apply a counting argument to obtain the following claim.
\begin{claim}
    Assuming $n\geq 5$ even, $k = n/2 - 1$, we have at least one of the following:
    \begin{enumerate}
        \item there is at least one vertex $v \notin I$ such that $v \notin L$ -- meaning $m(v) = 0$, or
        \item there are at least four vertices that occur in $L$ with multiplicity exactly 1.
    \end{enumerate}
    \label{claim:proof-num-intervals}
\end{claim}
\begin{proof}
    We show that when 2.\ does not hold, then 1.\ holds.

    If Case 2.\ does not hold, there are at most three labels with multiplicity exactly 1 in $L$.
    We now give an upper bound on the size of the set of distinct labels, that is $L^* := \{v \in V\setminus I \mid m(v) > 0\}$.
    Since the number of edges being labelled is fixed, $L^*$ is largest if each label occurs as rarely as possible, i.e., when there are three labels with multiplicity 1 and all other labels occur with multiplicity 2.
    Note that there are exactly $2k$ edges in $E_I$ so the total number of labels is $2k$.
    Thus, the number of labels occurring with multiplicity 2 can be at most $(2k-3)/2 = k-3/2$.
    Thus, $|L^*|\leq k-3/2+3 = k+3/2$, and since this number needs to be an integer we get $|L^*| \leq k+1$.
    Since there are $k$ intervals, at least $n-k-|L^*|$ vertices are neither an interval nor in $L^*$.
    As we assume $k = n/2 - 1$, this implies that $n-k-|L^*| \geq n-k-(k+1) = 1$ and thus, there must be at least one bicomponent not occurring in $L$, as claimed by Case~1.
\end{proof}
We next show that in each of these cases $k = n/2-1$ implies that $\mu(G) > n$.
First consider Case~1.
\begin{lemma}
    Assume $n\geq 5$ even, $k = n/2 - 1$.
    If Case~1.\ of \Cref{claim:proof-num-intervals} applies, we have $\mu(G) \geq n+1$.
    \label{lem:num-intervals-case1}
\end{lemma}
\begin{proof}
    In this case there is a bicomponent $v$ that is not adjacent to any interval.
    Therefore, consider the cut given by the set of intervals $I$.
    We already argued above that the size of this cut is $2k$.
    Now consider the cut $S = I \cup \{v\}$ obtained by adding $v$ to $I$.
    In addition to the cut edges incident to $I$, all incident edges of $v$ are cut edges in $S$, which implies $\mu(G) \geq 2k + 3 = n - 2 + 3 = n+1$.
\end{proof}

We proceed with Case~2, which will get significantly more  technical.
\begin{lemma}
    Assume $n\geq 5$ even, $k = n/2 - 1$.
    If Case~2.\ of \Cref{claim:proof-num-intervals} applies, we have $\mu(G) \geq n+1$.
    \label{lem:num-intervals-case2}
\end{lemma}
\begin{proof}
    Our strategy is similar to Case~1, where we used the cut given by $I$ and augmented it to get a cut of a size contradicting $\mu(G) \leq n$.
    We first observe that we can assume that there is no vertex $v \notin I$ such that $v \notin L$, as then Case~1.\ applies.
    
    To make the arguments easier to comprehend, we modify the graph $G$ by adding an edge between the two vertices of degree 3 (such an edge may already be present, in which case we simply double it).
    After this modification every vertex $v \in V(G)$ has a degree of $\deg(v) \in \{2, 4\}$.
    However, the max-cut could increase by that edge.
    Therefore, we will show that the modified graph has a cut of size at least $n+2$, which then implies this lemma.
    
    By Case~2.\ there are four vertices with multiplicity exactly 1 in $L$, call them $x_1, x_2, x_3, x_4$.
    We use the set $X = \{x_1, x_2, x_3, x_4\}$ to refer to these vertices.
    Observe that when two vertices in $X$ are non-adjacent, say $\{x_1, x_2\} \notin E(G)$, then we can augment the cut given by $I$, by adding both $x_1$ and $x_2$ to this cut.
    This way, the cut $I \cup \{x_1, x_2\}$ will have $2k - 2 + 6 = 2k + 4 = n+2$ cut edges, as claimed.
    This is also visualised in \Cref{fig:walk_graph_maxcut1}.
    \begin{figure}
        \centering
        \includegraphics{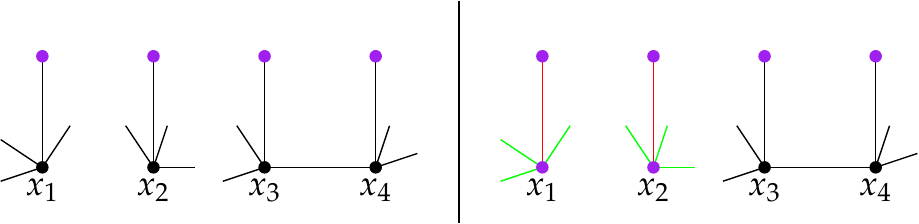}
        \caption{Augmenting the cut when $x_1$ and $x_2$ are non-adjacent. Both the situations before augmenting the cut given by $I$ (left) and after the augmentation with the cut $S = I \cup \{x_1, x_2\}$ (right) are shown.
        The vertices in the cut are coloured purple. Edges that are lost in the cut are coloured red and edges gained in the cut are coloured green.
        Hence, in $S$ we gain six additional cut edges, while loosing two, resulting in a total of four additional cut edges.}
        \label{fig:walk_graph_maxcut1}
    \end{figure}
    
    If there are more than four vertices with multiplicity exactly 1, we can augment the cut in the exact same way: Since each of these vertices must also be adjacent to an interval, they cannot all be pairwise adjacent.
    Therefore, we can now assume that there are exactly four vertices with multiplicity at most 1, and they are pairwise adjacent.
    
    We can also see that $n > 8$. If $n=6$ there is only a single possible graph (two intervals each adjacent to exactly two vertices of $X$), but this graph has a cut of size $8=n+2$. If $n=8$ it is impossible to accommodate all possible edges without either having neighbouring intervals or a disconnected graph.

    Moreover, we can see that there is no vertex with multiplicity at least 3.
    Indeed, assuming there is such a vertex $v \notin I$ with $m(v) \geq 3$ there are at most $2k - 4 -3$ edges to label with the remaining $n-k-4-1$ vertices not in $X\cup\{v\}$.
    But then the remaining average multiplicity is at most $(2k-7)/(n-k-5) = (n-9)/(n/2 - 4) = 2 - 2/(n-8) < 2$ for $n > 8$.
    Hence, there would be at least one other vertex with multiplicity at most 1, which we assumed not to be the case.
    We can thus assume that there are exactly four vertices with multiplicity 1 and all other bicomponents have multiplicity exactly 2.

    Let $Y = \{y_1, y_2, y_3, y_4\}$ be the intervals adjacent to $x_1, x_2, x_3, x_4$, respectively.
    First assume $|Y| < 4$, i.e., without loss of generality $y_1 = y_2$.
    Consider the cut $S = I \setminus \{y_1\} \cup \{x_1, x_2\}$.
    As shown in \Cref{fig:walk_graph_maxcut2} this augments the cut given by $I$ by four edges, implying that $\mu(G) \geq 2k + 4 = n+2$.
    \begin{figure}[t]
        \centering
        \includegraphics{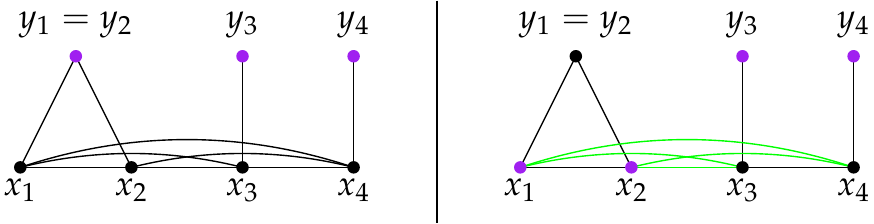}
        \caption{Augmenting the cut when $x_1,x_2,x_3,x_4$ induce a complete subgraph and $y_1 = y_2$. The left side depicts the cut $I$ (purple vertices correspond to vertices in the cut) and the right side depicts the cut $S = I \setminus \{y_1\} \cup \{x_1,x_2\}$. There are four additional cut edges in $S$ (coloured green).}
        \label{fig:walk_graph_maxcut2}
    \end{figure}

    Next, assume $|Y| = 4$.
    Let $Z = \{z_1, z_2, z_3, z_4\}$ be the bicomponents adjacent to $y_1, y_2, y_3, y_4$, such that $z_i \neq x_i$.
    If $|Z| \neq 4$, we have without loss of generality $z_1 = z_2$.
    Consider the cut $S = I \setminus \{y_1, y_2\} \cup \{x_1, x_2, z_1\}$.
    As shown in \Cref{fig:walk_graph_maxcut3} this augments the cut given by $I$ by six edges, implying that $\mu(G) \geq 2k + 6 \geq n+2$.
    \begin{figure}[b]
        \centering
        \includegraphics{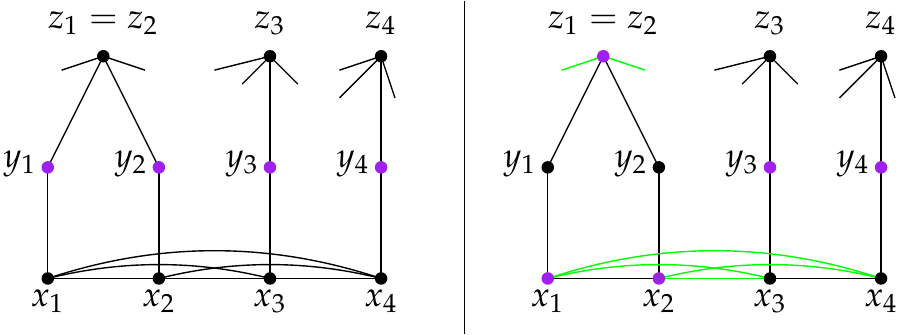}
        \caption{Augmenting the cut when $z_1 = z_2$. As usual, the left side depicts the cut $I$ (with purple vertices corresponding to vertices in the cut). The right side depicts the cut $S = I \setminus \{y_1, y_2\} \cup \{x_1, x_2, z_1\}$. Note that the multiplicity of $z_1$ is 2, so $z_1$ is not adjacent to an interval other than $y_1, y_2$. There are six additional cut edges in $S$.}
        \label{fig:walk_graph_maxcut3}
    \end{figure}

    Finally, consider the case where $|Y|=4$ and $|Z|=4$.
    Then since for all $z_i \in Z$ we have $\deg(z_i) = 4$ and $m(z_i) = 2$, there must be two non-adjacent vertices in $Z$, without loss of generality let them be $z_1$ and $z_2$.
    Then the cut given by $S = I \setminus \{y_1, y_2\} \cup \{x_1, x_2, z_1, z_2\}$ augments the cut given by $I$ by six edges as shown in \Cref{fig:walk_graph_maxcut4}.
    This implies that $\mu(G) \geq 2k + 6 \geq n+2$.
    \begin{figure}
        \centering
        \includegraphics{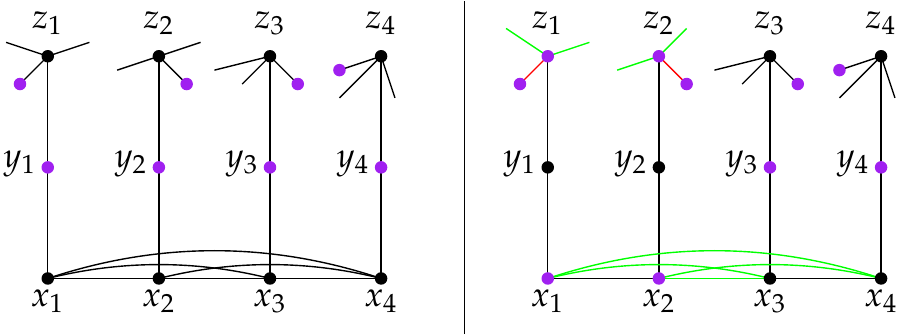}
        \caption{Augmenting the cut in the case when $|Y|=4$ and $|Z|=4$. Again, the left side depicts $I$ with purple vertices corresponding to vertices in the cut. The right side depicts the cut $S = I \setminus \{y_1, y_2\} \cup \{x_1, x_2, z_1, z_2\}$. Observe that the vertices $z_i$ are adjacent to some other interval different from $y_i$, because $z_i$ has multiplicity 2.
        This is depicted by the unlabelled purple vertices. We see that in $S$ we have eight additional cut edges (green) and two edges are lost in the cut (red), resulting in a total of six additional edges.}
        \label{fig:walk_graph_maxcut4}
    \end{figure}

    Observe that this case distinction is conclusive, where in any case we have shown that $\mu(G) \geq n+2$, exactly as required.
\end{proof}

Combining \Cref{claim:proof-num-intervals}, \Cref{lem:num-intervals-case1} and \Cref{lem:num-intervals-case2} we get the following corollary.
\begin{corollary}
    Let $G$ be an irreducible walk graph on $n$ vertices for some even $n \geq 5$.
    Then the number of vertices of degree 2, denoted by $k$, is exactly $k= n / 2$.
    \label{lem:number-intervals-irreducible}
\end{corollary}

Now that we know that the number of intervals is correct, the next step is to show that each bicomponent is adjacent to exactly two intervals in the walk graph.
\begin{lemma}
    Let $G$ be an irreducible walk graph on $n$ vertices for some even $n \geq 5$.
    Then every bicomponent is adjacent to exactly two intervals, and there is no double edge between any bicomponent and any interval.
    \label{lem:bicomponents-exactly-two-adjacent-intervals}
\end{lemma}
\begin{proof}
    First note that every bicomponent is adjacent to at least one interval.
    Otherwise, use this bicomponent to augment the cut $I$ to get a cut of size at least $2k+3 > n$.
    
    Similarly, if there is a bicomponent $u$ that has exactly one edge to intervals, say to interval $v$, we can use this bicomponent to augment the cut $I$ as well: From \Cref{lem:number-intervals-irreducible} we know that the size of the cut given by $I$ is exactly $n$. Adding the bicomponent $u$ of degree at least 3 that is only adjacent to one interval $v$ will increase the cut size by at least one.

    We next prove that there is no bicomponent with more than two edges to intervals (not even parallel edges). Since there are exactly $k=n/2$ intervals that have a total degree of $n$, if any bicomponent has three or more edges to intervals, then at least one bicomponent has at most one edge to intervals, which we have already excluded.

    It only remains to show that for every bicomponent the two edges to intervals go to two distinct intervals. If a bicomponent $u$ has two edges to some interval $v$ but otherwise is only adjacent to other bicomponents, the cut given by $I$ can be made strictly larger by replacing $v$ with $u$. We thus conclude the lemma.
\end{proof}

We thus know that every interval is adjacent to exactly two distinct vertices, which must be bicomponents by property c) of \Cref{lem:irreducibleconditions}. Since $N_n$ roughly consists of a cycle and a path connecting every other vertex in that cycle, these adjacent bicomponents should be connected.
\begin{lemma}
    Let $G$ be an irreducible walk graph on $n$ vertices for some even $n \geq 5$.
    For any vertex $v$ with $\deg(v) = 2$, let $v_1$ and $v_2$ be its two adjacent vertices.
    If $\deg(v_1) = 4$ or $\deg(v_2) =4$, then $v_1$ and $v_2$ are adjacent in $G$.
    \label{lem:deg2-adjacent-are-adjacent}
\end{lemma}
\begin{proof}
    We proceed by contradiction.
    Assume there is an interval $v$ with two non-adjacent neighbours $v_1$ and $v_2$, such that either $v_1$ or $v_2$ has degree 4.
    Consider the cut given by the set of intervals $I$.
    Its size is $2k = n$ by \Cref{lem:number-intervals-irreducible}.
    Now modify this cut by removing $v$ but inserting $v_1$ and $v_2$, so we get the cut $S = I \setminus \{v\}\cup \{v_1, v_2\}$.
    In this new cut the edges $\{v,v_1\}$ and $\{v,v_2\}$ remain in the cut.
    From \Cref{lem:bicomponents-exactly-two-adjacent-intervals}, we know that both $v_1$ and $v_2$ each have exactly one edge to other intervals, call them $u_1 \neq v$ and $u_2 \neq v$ respectively.
    Hence, in $S$ the edges $\{u_1, v_1\}$ and $\{u_2, v_2\}$ are removed.
    However, since $v_1$ and $v_2$ are not adjacent all other edges adjacent to them are added to the cut.
    Since at least one of them has degree 4, we gain at least three cut edges.
    Hence, $S$ has larger size than the cut given by $I$.
    This is a contradiction to $\mu(G) \leq n$.
    \begin{figure}
        \centering
        \includegraphics{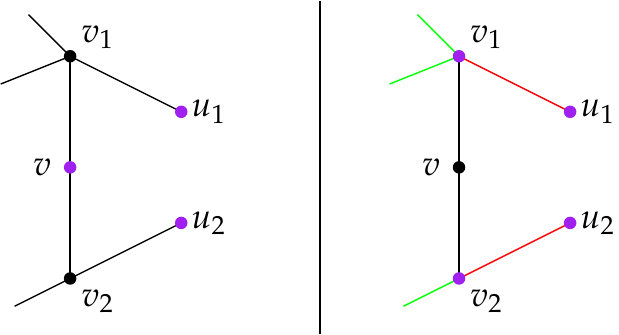}
        \caption{The augmentation from the proof of \Cref{lem:deg2-adjacent-are-adjacent}. The left side shows the cut $I$ (purple vertices correspond to vertices in the cut) and the right side shows the cut given by $I \setminus \{v\} \cup \{v_1,v_2\}$. Since $v_1$ and $v_2$ are non-adjacent all edges incident to $v_1$ and $v_2$ are gained as cut edges, except the two edges incident to the intervals $u_1$ and $u_2$. As at least one of $v_1, v_2$ is of degree 4 -- in this case $v_1$, we have a total of at least one additional cut edge.}
        \label{fig:walk_graph_maxcut5}  
    \end{figure}
    For an illustration of this augmentation see \Cref{fig:walk_graph_maxcut5}.
\end{proof}

Now we are ready to prove \Cref{prop:irreducible-graphs-isomorphic} for even $n$.
\begin{proof}[Proof of \Cref{prop:irreducible-graphs-isomorphic}, for even $n$]
    We build a sequence of distinct vertices {$u_1, ..., u_n$} that construct a cycle in $G$ such that for any odd $i < n-1$, there is an edge between $u_i$ and $u_{i+2}$.
    This is exactly the graph $N_n$ and since both $G$ and $N_n$ have the same number of vertices and edges, this then implies that $G$ is isomorphic to $N_n$.

    To start, let $u_1$ be an arbitrary vertex of degree $\deg(u_1) =3$, existing by properties a) and b) of \Cref{lem:irreducibleconditions}.
    By \Cref{lem:bicomponents-exactly-two-adjacent-intervals}, the vertex $u_1$ is adjacent to exactly two intervals $i_1,i_2$. Each of these intervals is adjacent to another bicomponent. One of the intervals must be adjacent to a bicomponent of degree 4, since otherwise both intervals would be adjacent to the two odd-degree vertices $u_1,u'$. Then the cut $I\setminus\{i_1,i_2\}\cup\{u_1,u'\}$ would be strictly larger than the cut given by $I$, which contradicts property d). Thus let $u_2\in\{i_1,i_2\}$ be an interval adjacent to $u_1$ and to a degree 4 vertex $u_3$.
    Note that $u_1$ and $u_3$ are adjacent by \Cref{lem:deg2-adjacent-are-adjacent}.
    We can now iteratively construct the sequence for every $i > 3$.
    To do so, we maintain the following invariant, which we have proven for $i \leq 3$.
    \begin{enumerate}[(1)]
        \item The vertices $u_1,\ldots,u_i$ are distinct,
        \item for all $j < i$ the vertices $u_j$ and $u_{j+1}$ are adjacent in $G$,
        \item for all even $j \leq i$ we have $\deg(v_j) = 2$, and
        \item for all odd $j \leq i-2$ we have $\deg(v_j) \geq 3$, with $\deg(v_1) = 3$, and $v_{j}$ is adjacent to $v_{j+2}$.
    \end{enumerate}
    To construct the sequence, we assume that for any $3 \leq i < n$ we have the subsequence $u_1, \dots, u_i$ such that the invariant holds for $i$ and we augment this sequence to $u_1, \dots, u_i, u_{i+1}$ that still maintains the invariant.
    Therefore, let $3 \leq i < n$ be arbitrary.
    Consider the following cases.

    \paragraph{Case 1: \texorpdfstring{$i$}{i} is even} Then $\deg(u_i) = 2$ (by (3)) and one of its adjacent vertices is $u_{i-1}$ (by (2)).
    Let $u_{i+1}$ be the other adjacent vertex to $u_i$ -- since $u_i$ cannot be adjacent to any interval we have $\deg(u_{i+1}) \geq 3$.

    To show that $u_{i+1} \neq u_j$ for any $j < i$ note that $j$ cannot be even as no intervals are adjacent.
    Moreover, for any odd $1 < j < i-1$, $u_{j}$ already has four edges, since it is adjacent to $u_{j-2}, u_{j-1}, u_{j+1}, u_{j+2}$ -- this is implied by the invariant. Since $u_i$ is not one of these vertices (using (1) and $j \leq i-3$) and $u_i$ and $u_{i+1}$ must be adjacent, $u_{i+1} \neq u_j$.
    Hence, the only remaining possibilities are $j = 1$ or $j=i-1$.
    
    If $u_{i+1}=u_1$ the vertices $u_1,\ldots,u_i$ must form a connected component of $G$. However, since $i<n$ this would imply that $G$ is not connected, and thus not semi-Eulerian.
    If $u_{i+1}=u_{i-1}$ we would have that $u_i$ has only one adjacent bicomponent, contradicting \Cref{lem:bicomponents-exactly-two-adjacent-intervals}.
    
    To show that the invariant still holds, it only remains to show that $u_{i-1}$ and $u_{i+1}$ are adjacent.
    Since neither $u_{i-1}$ nor $u_{i+1}$ are equal to $u_1$, at least one of them has degree 4. Thus the two vertices must be adjacent by \Cref{lem:deg2-adjacent-are-adjacent}.

    \paragraph{Case 2: \texorpdfstring{$i$}{i} is odd} Then $u_i$ is adjacent to the interval $u_{i-1}$ (by (2) and (3)).
    Moreover, $u_i$ is adjacent to some other interval by \Cref{lem:bicomponents-exactly-two-adjacent-intervals}, let $u_{i+1}$ be that vertex.
    The only non-trivial part of the invariant is to show $u_{i+1} \neq u_j$ for every $j < i$. However, every interval in $u_1,\ldots,u_i$ already has both of its edges accounted for, thus $u_{i+1}$ must be a new vertex.

    \paragraph{Conclusion} It remains to prove that the invariant indeed implies that the sequence $u_1, \dots, u_n$ forms a cycle of length $n$ and that for any odd $i< n-1$, $u_i$ and $u_{i+2}$ are adjacent.
    The latter is directly implied by (4) of the invariant.
    Moreover, for $i < n$ condition (2) implies that $u_i$ and $u_{i+1}$ are adjacent, so we need to show that $u_n$ and $u_1$ are adjacent which then implies that $u_1, \dots, u_n$ forms a cycle.
    The vertex $u_n$ has two adjacent bicomponents, one of which is $u_{n-1}$.
    Since in the sequence all the vertices $u_j$ for odd $1 < j < n-1$ already have all four edges accounted for by (2),(3),(4) and \Cref{lem:deg2-adjacent-are-adjacent}, the only target of $u_n$'s remaining edge can be $u_1$.
    Hence, $u_n$ and $u_1$ are adjacent implying that $u_1, \dots, u_n$ forms a cycle.

    Since $G$ neither contains more vertices nor more edges than $N_n$, we can conclude that the graphs are isomorphic.
\end{proof}

\subsection{Odd \texorpdfstring{$n$}{n}}
To see that \Cref{prop:irreducible-graphs-isomorphic} also holds for odd $n$, we will use that in this case the two vertices of degree 3 must be adjacent in the walk graph.
\begin{lemma}
    Let $G$ be an irreducible walk graph on $n$ vertices for some odd $n \geq 5$. Then the two vertices of degree 3 are adjacent.
    \label{lem:deg3-adjacent-n-odd}
\end{lemma}
Before we prove this lemma, let us see how this can be used to show \Cref{prop:irreducible-graphs-isomorphic} for odd~$n$.

\begin{proof}[Proof of \Cref{prop:irreducible-graphs-isomorphic}, for odd $n$]
We reduce the odd case to the even case.
In particular, by \Cref{lem:deg3-adjacent-n-odd} there is an edge between the two vertices of degree 3 in the walk graph $G$.
We subdivide this edge with a new vertex and obtain a graph $G'$ on $n+1$ vertices.
The goal is to show that $G'$ satisfies the conditions a)--d) of an irreducible walk graph (\Cref{lem:irreducibleconditions}).
Note that a)--c) trivially hold in $G'$.

Therefore, we only need to show that $\mu(G') \leq n+1$.
Towards a contradiction assume $\mu(G') > n+1$, and let $S$ be a maximum cut in $G'$.
Consider this same cut in $G$.
Let $v_1$ and $v_2$ be the two vertices of degree 3.

If $v_1$ and $v_2$ are on opposite sides of the cut, exactly one of the edges making up $\{v_1,v_2\}$ in $G'$ is a cut edge. Similarly, $\{v_1,v_2\}$ is a cut edge in $G$. Thus the cut in $G$ will have size exactly $\mu(G)=\mu(G')>n$, a contradiction.

For the case that $v_1$ and $v_2$ are on the same side of the cut, note that every cut edge in a cut of $G'$ corresponds to a cut in the necklace underlying $G'$.
Therefore, since $v_1$ and $v_2$ correspond to the first and last component of the necklace they can only be on the same side of the cut if there is an even number of cut edges.
Since $n$ is odd and we assume $\mu(G') > n+1$ this implies that $\mu(G') \geq n+3$.
As when going from $G'$ to $G$ the cut size decreases by at most two edges we conclude that $\mu(G) \geq \mu(G)-2 \geq n+3-2 > n$, also a contradiction.

We have shown that $G'$ is an irreducible walk graph on $n+1$ --- an even number --- vertices.
Thus, $G'$ is isomorphic to $N_{n+1}$ by the even case of \Cref{prop:irreducible-graphs-isomorphic} proven earlier.
We now observe that $G'$ is obtained by replacing an edge between the two odd-degree vertices in $G$ by an interval and there is only one such interval in $N_{n+1}$ adjacent to both odd-degree vertices. Thus $G$ must be isomorphic to $N_{n+1}$ with this interval replaced by an edge. This yields exactly the graph $N_{n}$.
Hence, we can conclude that \Cref{prop:irreducible-graphs-isomorphic} also holds in the odd case.    
\end{proof}

It remains to show \Cref{lem:deg3-adjacent-n-odd}.
\begin{proof}[Proof of \Cref{lem:deg3-adjacent-n-odd}]
    We prove the lemma by contradiction, so assume that the two degree three vertices are non-adjacent.
    To show the lemma, we distinguish by the number of intervals $k$ in $G$ -- note that $k \in \{\frac{n-3}{2},\frac{n-1}{2}\}$, by \Cref{lem:num-intervals-irreducible-almost} for odd $n$.
    We prove the lemma for both cases separately.

    In both cases we use a very similar strategy as in the proof of \Cref{lem:number-intervals-irreducible}, where the assumption on the number of intervals can be used to derive two cases to distinguish the walk graph.
    For each of these two cases we show that the cut given by the set of intervals, $I$, can be augmented to a cut of size larger than $n$.
    In fact, we use the same terminology for labelling the incident edges of the intervals as in the proof of \Cref{claim:proof-num-intervals}.
    Recall that for an edge $e = \{u,v\}$ incident to an interval $u$ and a bicomponent $v$, we label this edge with $l(e) := v$.
    The multiset of labels is given by $L := \{v^{m(v)} \mid \deg(v) \geq 3\}$, where the multiplicity of a bicomponent $v$ is given by $m(v) := |\{e \in E_I \mid l(e) = v\}|$, and $E_I$ is the set of edges incident to intervals.

    \paragraph{Case 1: \texorpdfstring{$k = (n-3)/2$}{k=(n-3)/2}}
    We need the following claim.
    \begin{claim}
        If $k = \frac{n-3}{2}$ we have at least one of the following:
        \begin{enumerate}
            \item There is at least one vertex of degree 4 with multiplicity 0 in $L$, or
            \item one of the vertices of degree 3 has multiplicity 0 and the other multiplicity at most 1, or
            \item there are at least four vertices of degree 4 with multiplicity 1 in $L$.
        \end{enumerate}
        \label{claim:n-odd-1}   
    \end{claim}
    \begin{proof}[Proof of \Cref{claim:n-odd-1}]
        Consider the total multiplicity $T$ of the vertices of degree 3.
        That is, for $u$ and $v$ being the two degree 3 vertices we have $T = m(u) + m(v)$.
        If $T < 2$ we must have that condition 2.\ holds.
        So assume $T \geq 2$.
        The remaining total multiplicity shared among the degree 4 vertices is then $2k - T \leq n-5$.
        Note that there are $n-k-2 = (n-1)/2$ vertices of degree 4.
        However, if both 1.\ and 3.\ do not hold, the multiplicity used by vertices of degree at least 4 is at least $1\cdot 3 + 2\cdot((n-1)/2-3) = n-4 > n-5$.
        Hence the claim follows.
    \end{proof}
    
    Now notice that \Cref{claim:n-odd-1} implies the following.
    If condition 1.\ holds we can augment the cut given by $I$ with that vertex of multiplicity 0 to get a cut of size $2k + 4 = (n-3)+4 > n$.

    If condition 2.\ holds we can add both vertices of degree 3 vertices to the cut and since they are non-adjacent and have total multiplicity at most 1, we add at least 5 edges to the cut and remove at most 1 edge. Thus the cut size increases by at least four in this case as well.

    Finally, if condition 3.\ but none of the other conditions hold we can increase the cut size by at least four edges in a very similar way to the proof of \Cref{lem:num-intervals-case2}. We repeat the proof in the following claim, since there are some small subtle changes.

    \begin{claim}\label{claim:repeatoflemma}
        Assuming that condition 3.\ is the only condition of \Cref{claim:n-odd-1} that holds and the vertices of degree 3 are not adjacent, we must have have $\mu(G)>n$.
    \end{claim}
    \begin{proof}[Proof of \Cref{claim:repeatoflemma}]
    We first observe that we can assume that there is no degree 4 vertex of multiplicity 0, as then condition~1.\ applies.

    We use the set $X = \{x_1, x_2, x_3, x_4\}$ to refer to the degree 4 multiplicity 1 vertices given by condition 3.
    Observe that when two vertices in $X$ are non-adjacent, say $\{x_1, x_2\} \notin E(G)$, then we can augment the cut given by $I$, by adding both $x_1$ and $x_2$ to this cut.
    This way, the cut $I \cup \{x_1, x_2\}$ will have $2k - 2 + 6 = 2k + 4 = n+1$ cut edges.
    
    If there are any other degree 4 vertices with multiplicity exactly 1, we can augment the cut in the exact same way: Since each of these vertices must also be adjacent to an interval, they cannot all be pairwise adjacent.
    Therefore, we can now assume that there are exactly four vertices with degree 4 and multiplicity at most 1, and they are pairwise adjacent.

    Moreover, we can see that there is no degree 4 vertex with multiplicity at least 3. Assuming there is a degree 4 vertex $v$ with $m(v)\geq 3$, then $v,x_1,\ldots,x_4$, and the two degree 3 vertices together would use at least $3+4+2$ labels (since condition 2.\ is not satisfied). Thus there are at most $2k - 9$ edges to label with the remaining $n-k-4-1-2$ degree 4 vertices not in $X\cup\{v\}$.
    But then the remaining average multiplicity is at most $(2k-9)/(n-k-7) = (n-12)/(n/2 - 11/2) = 2 - 2/(n-11) < 2$ for $n > 11$ (note that if we had $n\leq 11$ we would have at most four intervals which would provide at most 8 total multiplicity, which is not enough).
    Hence, there would be at least one other degree 4 vertex with multiplicity at most 1, which we assumed not to be the case.
    We can thus assume that there are exactly four degree 4 vertices with multiplicity 1 and all other degree 4 vertices have multiplicity exactly 2.

    Similarly, if the degree 3 vertices together had multiplicity at least 3, $X$ and the two degree 3 vertices would use at least $4+3$ labels. Thus there are at most $2k-7$ labels left for the remaining $n-k-4-2$ degree 4 vertices. Then the remaining average multiplicity is at most
    $(2k-7)/(n-k-6) = (n-10)/(n/2 - 9/2) = 2- 2/(n-9)<2$ for $n>9$ (note again that with $n\leq 9$ we had at most three intervals, providing only 6 total multiplicity, which is not enough). Thus there would again be another degree 4 vertex with multiplicity at most 1. We thus conclude that both degree 3 vertices together have total multiplicity exactly 2.

    Let $Y = \{y_1, y_2, y_3, y_4\}$ be the intervals adjacent to $x_1, x_2, x_3, x_4$, respectively.
    First assume $|Y| < 4$, i.e., without loss of generality $y_1 = y_2$.
    Consider the cut $S = I \setminus \{y_1\} \cup \{x_1, x_2\}$.
    This augments the cut given by $I$ by four edges, implying that $\mu(G) \geq 2k + 4 = n+2$.

    Next, assume $|Y| = 4$.
    Let $Z = \{z_1, z_2, z_3, z_4\}$ be the bicomponents adjacent to $y_1, y_2, y_3, y_4$, such that $z_i \neq x_i$.
    If $|Z| \neq 4$, we have without loss of generality $z_1 = z_2$.
    Consider the cut $S = I \setminus \{y_1, y_2\} \cup \{x_1, x_2, z_1\}$.
    Since no vertex has multiplicity more than 2, this augments the cut given by $I$ by at least five edges, implying that $\mu(G) \geq 2k + 5 \geq n+2$.

    Finally, consider the case where $|Y|=4$ and $|Z|=4$.
    Then since for at least two $z_i \in Z$ we have $\deg(z_i) = 4$ and $m(z_i) = 2$, there must be two non-adjacent vertices in $Z$, without loss of generality let them be $z_1$ and $z_2$.
    Then the cut given by $S = I \setminus \{y_1, y_2\} \cup \{x_1, x_2, z_1, z_2\}$ augments the cut given by $I$ by at least four edges.
    This implies that $\mu(G) \geq 2k + 4 \geq n + 1$.
    \end{proof}
        
    In all three cases of \Cref{claim:n-odd-1} we thus get a contradiction to the assumption that $\mu(G)\leq n$.
    Hence, if $k = (n-3)/2$ the two degree 3 vertices must be adjacent.

    \paragraph{Case 2: \texorpdfstring{$k = (n-1)/2$}{k=(n-1)/2}}
    Similar to the previous case, we need the following claim.
    \begin{claim}
        If $k = \frac{n-1}{2}$ we have at least one of the following:
        \begin{enumerate}
            \item There is at least one vertex with multiplicity 0 in $L$, or
            \item there are at least two vertices with multiplicity 1 in $L$.
        \end{enumerate}
        \label{claim:n-odd-2} 
    \end{claim}
    \begin{proof}[Proof of \Cref{claim:n-odd-2}]
        Assume 2.\ does not hold.
        Then there is at most one vertex of multiplicity 1.
        The remaining multiplicity is then given by $2k-1$, hence the number of vertices with multiplicity 2 or larger is at most $\frac{2k-1}{2} = k - 1/2$.
        Thus the number of distinct labels is at most $|L^*| \leq k-1/2+1$, hence $|L^*| \leq k$.
        Therefore, there are at least $n-k-|L^*|$ vertices with multiplicity 0.
        The claim follows observing that $n-k-|L^*| \geq n - \frac{n-1}{2} - \frac{n-1}{2} = 1$.
    \end{proof}
    If condition 1.\ of the claim holds, augmenting the cut $I$ by the vertex of multiplicity 0 increases the cut by at least three edges, implying $\mu(G)\geq 2k+3 = n+2 > n$.
    
    If condition 2.\ holds and one of the vertices with multiplicity 1 is of degree 4, we can augment the cut $I$ by inserting that vertex, increasing the cut size by at least two as well.
    If both of the vertices with multiplicity 1 are of degree 3, since we assume they are non-adjacent, we can augment the cut $I$ by inserting these vertices.
    Doing so will increase the cut by at least two edges as well.

    We have shown that also in the case of $k = \frac{n-1}{2}$ the two degree 3 vertices are adjacent.
\end{proof}

\section{Proof of \texorpdfstring{\Cref{lem:treewidth-enhancement}}{Lemma \ref{lem:treewidth-enhancement}}}\label{app:enhancedgraph}
\begin{proof}[Proof of \Cref{lem:treewidth-enhancement}]
    Let $(\mathcal{T}, \mathcal{B})$ be a tree decomposition of $\mathcal{G}$ of width $\tw(\mathcal{G})$.
    We now construct a tree decomposition of $G$ of width $w \cdot (\tw(\mathcal{G})+1) -1 $, on the same tree $\mathcal{T}$.
    Consider the function $B: V(\mathcal{T}) \to 2^{V}$ defined as $B(x) := \bigcup_{X_i \in \mathcal{B}(x)} X_i$ for all $x \in V(\mathcal{T})$.
    Since $|\mathcal{B}(x)| \leq \tw(\mathcal{G}) + 1$ for all $x \in V(\mathcal{T})$ and $|X_i| \leq w$ for all $X_i \in \mathcal{X}$ it follows that the width of this decomposition is at most $w \cdot (\tw(\mathcal{G})+1) - 1$.
    It remains to show that $(\mathcal{T},B)$ is indeed a valid tree decomposition of $G$.
    In order to do so, we need to show that the conditions 1--3 from \Cref{def:treewidth} are satisfied.

    The first two conditions are straightforward.
    For the first condition we have 
    \[\bigcup_{x \in V(\mathcal{T})}B(x) = \bigcup_{x \in V(\mathcal{T})}\bigcup_{X_i \in \mathcal{B}(x)}X_i = \bigcup_{X_i \in \mathcal{X}} X_i = V,\]
    first using that $(\mathcal{T, B})$ is a valid tree decomposition of $\mathcal{G}$ and then that $\mathcal{G}$ is an enhanced graph of $G$.
    
    Similarly, for the second condition consider an edge $e\in E$. If $e\subseteq X_i$ for some $X_i \in \mathcal{X}$, then any $x \in V(\mathcal{T})$ with $X_i \in \mathcal{B}(x)$ has $e \subseteq B(x)$.
    If otherwise $e = \{u,v\} \nsubseteq X_i$ for all $X_i \in \mathcal{X}$, we let $X_u,X_v \in \mathcal{X}$ be such that $u \in X_u$ and $v \in X_v$.
    Then by definition of an enhanced graph there must be an edge between $X_u$ and $X_v$ in $\mathcal{G}$.
    Since $(\mathcal{T}, \mathcal{B})$ is a tree decomposition of $\mathcal{G}$ there is an $x \in V(\mathcal{T})$ such that $\{X_u, X_v\} \subseteq \mathcal{B}(x)$, implying $\{u,v\} \subseteq B(x)$.

    It is left to prove that for all $v \in V$ the subgraph $\mathcal{T}[\{x \mid v \in B(x)\}]$ is connected.
    Note that $S:=\{x \mid v \in B(x)\} = \{x \mid v \in \bigcup_{X_i \in \mathcal{B}(x)} X_i\}$.
    Let $\mathcal{X}_v := \{X_i \in \mathcal{X} \mid v \in X_i\}$.
    Then the set $S$ is equal to $S = \{x \mid \mathcal{B}(x) \cap \mathcal{X}_v \neq \emptyset\}$.
    We can rewrite this further as $S = \bigcup_{X_i \in \mathcal{X}_v} \{x \mid X_i \in \mathcal{B}(x)\}$.

    To show the desired results we need to show that $\mathcal{T}[S]$ is connected.
    In order to do so, we show that for any $X_i, X_j \in \mathcal{X}_v$ we have $\{x \mid X_i \in \mathcal{B}(x)\} \cap \{x \mid X_j \in \mathcal{B}(x)\} \neq \emptyset$.
    To see why this is sufficient, notice that as $(\mathcal{T}, \mathcal{B})$ is a tree decomposition of $\mathcal{G}$, both $\mathcal{T}[\{x \mid X_i \in \mathcal{B}(x)\}]$ and $\mathcal{T}[\{x \mid X_j \in \mathcal{B}(x)\}]$ are connected.
    Thus, if these sets share an element the subgraph induced by the union of these sets must be connected. Thus, the subgraph induced by the union of all of these sets, i.e., $\mathcal{T}[S]$ is connected.

    Note that by definition both $X_i$ and $X_j$ contain $v$, meaning $X_i \cap X_j \neq \emptyset$.
    But then, by definition of an enhanced graph there is an edge between $X_i$ and $X_j$ in $\mathcal{G}$.
    Further using the second condition of the definition of a tree decomposition we can conclude that there must be an $x^* \in V(\mathcal{T})$ such that $\{X_i, X_j\} \subseteq \mathcal{B}(x^*)$.
    But then $x^* \in \{x \mid X_i \in \mathcal{B}(x) \}$ and $x^* \in \{x \mid X_j \in \mathcal{B}(x)\}$, and thus their intersection is non-empty.
\end{proof}

\section{Example of the Reduction}\label{app:example}
We show how the reduction from \Cref{thm:NPhardness} turns the following example  \textsc{e3-Sat} instance into a necklace.
\[\Phi = C_1 \land C_2, \;\text{where }\;C_1 = x \lor \overline{y} \lor z \;\text{ and }\;C_2 = x \lor \overline{z} \lor w\]
There are four variables, and the variable parts are given by:
\begin{align*}
    &P\, x_0^A\, x_T\, x_T\, x_0^B\, P\, x_0^A\, x_0^B\, P \\
    &P\, y_0^A\, y_T\, y_0^B\, P\, y_0^A\, y_0^B\, P \\
    &P\, z_0^A\, z_T\, z_T\, z_0^B\, P\, z_0^A\, z_0^B\, P \\
    &P\, w_0^A\, w_T\, w_0^B\, P\, w_0^A\, w_0^B\, P 
\end{align*}
The clause parts are given by, first for clause $C_1$
\begin{align*}
    &P\, x_1^A\, C_1\, x_1^B\, P\, x_1^A\, x_T\, x_1^B\, P \\
    &P\, y_1^A\, y_1^B\, P\, y_1^A\, y_T\, C_1\, y_1^B\, P \\
    &P\, z_1^A\, C_1\, z_1^B\, P\, z_1^A\, z_T\, z_1^B\, P
\end{align*}
and for $C_2$
\begin{align*}
    &P\, x_2^A\, C_2\, x_2^B\, P\, x_2^A\, x_T\, x_2^B\, P \\
    &P\, z_2^A\, z_2^B\, P\, z_2^A\, z_T\, C_2\, z_2^B\, P \\
    &P\, w_2^A\, C_2\, w_2^B\, P\, w_2^A\, w_T\, w_2^B\, P
\end{align*}
The enforcing part starts with the enforcement of the clauses by
\begin{align*}
    & P\, C_1\, C_1\, C_1\, S_1\, P \\
    & P\, C_2\, C_2\, C_2\, S_2\, P
\end{align*}
and for the variables
\begin{align*}
    & P\, x_T\, P\, x_T\, N\, x_T\, N\, x_T\, N\, S_x\\
    & P\, y_T\, N\, y_T\, N\, S_y\\
    & P\, z_T\, P\, z_T\, N\, z_T\, N\, z_T\, N\, S_z\\
    & P\, w_T\, N\, w_T\, N\, S_w
\end{align*}
Finally, the enforcement of $P$ and $N$ is given by
\begin{align*}
    P\, N   
\end{align*}
In conclusion, the necklace is given by the following string (line breaks and spacing are inserted for better readability, and to emphasise the structure of the construction). The beads in boldface are the beads that are cut in the cut corresponding to the satisfying assignment $x=y=z=0, w=1$:
\begin{align*}
    &a\, \mathbf{b\, a} \\~\\
    & P\, x_0^A\, x_T\, x_T\, x_0^B\, P\, \mathbf{x_0^A\, x_0^B}\, P  \quad P\, y_0^A\, y_T\, y_0^B\, P\, \mathbf{y_0^A\, y_0^B}\, P  \quad P\, z_0^A\, z_T\, z_T\, z_0^B\, P\, \mathbf{z_0^A\, z_0^B}\, P \\ 
    &P\, \mathbf{w_0^A}\, w_T\, \mathbf{w_0^B}\, P\, w_0^A\, w_0^B\, P \quad \\~\\
    &P\, x_1^A\, C_1\, x_1^B\, P\, \mathbf{x_1^A}\, x_T\, \mathbf{x_1^B}\, P \quad P\, y_1^A\, y_1^B\, P\, \mathbf{y_1^A}\, y_T\, C_1\, \mathbf{y_1^B}\, P \quad P\, z_1^A\, C_1\, z_1^B\, P\, \mathbf{z_1^A}\, z_T\, \mathbf{z_1^B}\, P \\
    &P\, x_2^A\, C_2\, x_2^B\, P\, \mathbf{x_2^A}\, x_T\, \mathbf{x_2^B}\, P \quad P\, z_2^A\, z_2^B\, P\, \mathbf{z_2^A}\, z_T\, C_2\, \mathbf{z_2^B}\, P  \quad P\, \mathbf{w_2^A}\, C_2\, \mathbf{w_2^B}\, P\, w_2^A\, w_T\, w_2^B\, P \\~\\
    &P\, \mathbf{C_1}\, C_1\, C_1\, \mathbf{S_1}\, P \quad P\, C_2\, \mathbf{C_2}\, C_2\, \mathbf{S_2}\, P\\~\\
    &P\, x_T\, P\, \mathbf{x_T}\, N\, x_T\, N\, x_T\, N\, \mathbf{S_x} \quad P\, \mathbf{y_T}\, N\, y_T\, N\, \mathbf{S_y} \quad P\, z_T\, P\, \mathbf{z_T}\, N\, z_T\, N\, z_T\, N\, \mathbf{S_z} \\
    &P\, \mathbf{w_T}\, N\, w_T\, N\, \mathbf{S_w} \\~\\
    &\mathbf{P\, N}    
\end{align*}
Moreover, $\alpha$ is given as 
\begin{align*}
    &\alpha(a) = 2, \alpha(b) = 1 \\
    &\alpha(P) = 41 = |P|, \alpha(N) = 1\\
    &\alpha(x_0^A) = \alpha(x_0^B) = \alpha(x_1^A) = \alpha(x_1^B) = \alpha(x_2^A) = \alpha(x_2^B) = 2\\
    &\alpha(y_0^A) = \alpha(y_0^B) = \alpha(y_1^A) = \alpha(y_1^B) =  2\\
    &\alpha(z_0^A) = \alpha(z_0^B) = \alpha(z_1^A) = \alpha(z_1^B) = \alpha(z_2^A) = \alpha(z_2^B) = 2\\
    &\alpha(w_0^A) = \alpha(w_0^B) = \alpha(w_2^A) = \alpha(w_2^B) = 2\\
    &\alpha(x_T) = 4 = \frac{|x_T|}{2},\alpha(y_T) = 2 = \frac{|y_T|}{2}, \alpha(z_T) = 4 = \frac{|z_T|}{2}, \alpha(w_T) = 2 = \frac{|w_T|}{2}\\
    &\alpha(C_1) = \alpha(C_2) = 3\\
    &\alpha(S_1) = \alpha(S_2) = \alpha(S_x) = \alpha(S_y) = \alpha(S_z) = \alpha(S_w) = 1.
\end{align*}

\end{document}